\documentclass[12pt]{article}
\usepackage{amsfonts,amssymb,amsthm,eucal,amsmath,verbatim,color,bbm}

\allowdisplaybreaks
\newtheorem{thm}{Theorem}
\newtheorem{cor}[thm]{Corollary}
\newtheorem{lemma}[thm]{Lemma}
\newtheorem{prop}[thm]{Proposition}

\newcommand{\R}{\mathbb{R}}
\newcommand{\Z}{\mathbb{Z}}

\newcommand{\E}{\mathbb{E}}

\DeclareMathOperator{\csch}{csch}
\newcommand{\inprod}[2]{\left\langle #1, #2 \right\rangle}

\newcommand{\F}{\mathcal{F}}

\renewcommand{\P}{\mathbb{P}}

\newcommand{\s}{\mathbb{S}}
\newcommand{\tr}{\mathrm{Tr\,}}

\newcommand{\1}{\mathbbm{1}}
\renewcommand{\O}{\mathcal{O}}

\begin{document}

\title
{Asymptotics of the mean-field Heisenberg
model}
\author{Kay Kirkpatrick\thanks{Partially supported by NSF grants
    OISE-0730136 and DMS-1106770} \;  and Elizabeth
  Meckes\thanks{Partially supported by the American Institute of
    Mathematics and NSF grant DMS-0852898} \\ \\
Department of Mathematics, University of Illinois at Urbana-Champaign\\
1409 W. Green Street, Urbana, IL 61801, USA \\ \\
Department of Mathematics, Case Western Reserve University \\
220 Yost Hall, Cleveland, OH 44106, USA}

\maketitle


\begin{abstract} 
We consider the mean-field classical Heisenberg model and obtain detailed information about the total spin of the system by studying the model on a complete graph and sending the number of vertices to infinity.
In particular, we obtain Cram\'er- and Sanov-type large deviations principles for the total spin and the empirical spin distribution and demonstrate a second-order phase transition in the Gibbs measures. We also study the asymptotics of the total spin throughout the phase transition using Stein's method, proving central limit theorems in the sub- and supercritical phases and a nonnormal limit theorem at the critical temperature.
\end{abstract}

\section{Introduction and summary of results}

For many models of statistical mechanics, understanding their physical behavior starts with understanding the behavior of the corresponding mean-field model--which not only suggests how the physical model behaves, but also can predict rather precisely the physical behavior in high dimensions. There are two main statistical mechanical models of ferromagnetism: the simpler and better-understood Ising model, and the more realistic and more challenging (classical) Heisenberg model, on a lattice of dimension $d$ with a spin $\sigma_i \in \s^2$ at each lattice site $i$, the spin configuration $\sigma \in (\s^2)^n$ having Hamiltonian energy in the absence of an external field (anisotropy):
$$H_n(\sigma) = - \sum_{i,j} J_{i,j}\inprod{\sigma_i}{\sigma_j}.$$
The nearest-neighbor Heisenberg model has constant interaction $J_{i,j} = J$ for nearest neighbors $i$ and $j$, and no interaction $J_{i,j}=0$ otherwise. The mean-field version of the Heisenberg model has an averaged interaction $J_{i,j} = \frac{1}{2n}$ for all $i,j$ and can be understood as sending the dimension $d \to \infty$ or considering the lattice to be a complete graph on $n$ vertices and sending $n \to \infty$.

The related quantum Heisenberg model has nearest-neighbor interactions of spin matrices, and in dimensions three and higher, there is a proof of a phase transition to long-range order for the antiferromagnetic quantum Heisenberg model ($J_{i,j} < 0$ for neighbors $i$ and $j$) \cite{DLS}.
Such a result has been thus far intractable for the ferromagnetic quantum Heisenberg model as well as the classical Heisenberg model, but the mean-field classical model is more amenable to rigorous analysis. In this article, we prove the existence of a second-order phase transition for the mean-field classical Heisenberg model, deriving a number of precise formulas and asymptotics for various physical quantities.

The simpler Ising model of ferromagnetism (the spins are $+1$ or
$-1$) is better understood, and our results parallel some recent
developments for the mean-field version of the Ising model, called the
Curie-Weiss model. It is believed that the Curie-Weiss model
accurately describes the Ising model in dimensions greater than four,
in the sense that they have the same critical exponents of various
physical quantities (e.g., total spin, free energy). The
total spin (appropriately normalized) in the
Curie-Weiss model was shown by Ellis and Newman \cite{EN} to have a Gaussian law in the non-critical
regimes and law that converges to the distribution with density
proportional to $e^{-x^4/12}$ at the critical temperature.  Recently,
it was shown by Chatterjee and Shao \cite{CS} that the total spin
at the critical temperature 
satisfies a Berry--Esseen type error bound of order $1/\sqrt{n}$ for
this non-central limit theorem. 
See also \cite{CET, EM} for analogous results on the Curie-Weiss-Potts
model with an arbitrary finite number  of discrete spins, and  \cite{BC} for other related models.

There are a few results known for the Heisenberg model. In dimensions one and two, the classical Heisenberg model with nearest-neighbor interactions has no symmetry breaking, i.e., there is no phase transition to asymmetric macrostates above a critical temperature, e.g., \cite{DS}. By contrast, the classical Heisenberg model with long-range interactions or anisotropy (an external magnetic field) has a phase transition in two dimensions and higher, see \cite{M} and references therein. And for classical, isotropic Heisenberg models with nearest-neighbor interactions in three dimensions and higher, the existence of a phase transition was shown using Fourier-based infrared bounds \cite{FSS}.

As for studying the mean-field classical Heisenberg model, the large-dimensional ($d \to \infty$) limit of the nearest-neighbor model on $\Z^d$, with spins in $\s^{2}$, has the critical inverse temperature $\beta_c=3$ \cite{KS}. Moreover,  the large-dimensional limit is known to be a good approximation for high-dimensional models (as in Landau's theory of second order phase transitions) in the sense that below the critical temperature, the total spin is zero for all $d$, and above the critical temperature, the total spin has the correct (non-zero) limit as $d \to \infty$.


The results  in this article for the isotropic classical mean-field (in the complete graph sense) Heisenberg model include:

\begin{itemize}
\item Large deviation principles (LDPs) for the total spin and the empirical spin distribution for each inverse temperature $\beta \in [0, \infty)$ with explicit rate functions and relative entropies. (Section 2)
\item Explicit formulas for the free energy and descriptions of the canonical macrostates and the corresponding second-order phase transition. (Section 2)
\item An LDP for the empirical spin distribution with respect to the microcanonical ensemble, which has fixed energy, descriptions of the microcanonical macrostates and their second-order transition. (Section 2)
\item A central limit theorem in the subcritical phase for the total spin with the usual CLT scaling $\sqrt{n}$, using Stein's method.  (Section 3)
\item A CLT in the supercritical phase for the total spin with a more complicated scaling, again using Stein's method.  (Section 4)
\item A nonnormal limit theorem for the total spin at the critical
  temperature, with  limiting density of the squared length
  proportional to $t^5 e^{-3ct^2}$  (Section 5), making use of a new abstract Stein's method result for the nonnormal approximation (Appendix).
\end{itemize}

\bigskip
\section{The model and large deviations results}

We consider the isotropic mean-field classical Heisenberg model on a finite complete graph $G_n$ with $n$ vertices. That is, at each site of the graph is a spin living in $\Omega = \s^2$, so the state space is $\Omega_n = (\s^2)^n$ with $P_n$ the $n$-fold product of the uniform probability measure on $\s^2$. For this model, the mean-field Hamiltonian energy $H_n : \Omega_n \to \R$ is: 
\[ H_n (\sigma) :=  -\frac{1}{2n} \sum_{i,j=1}^n\inprod{\sigma_i}{\sigma_j}. \]
The energy per particle is $h_n(\sigma)=\frac{1}{n}H_n(
\sigma)$, and the canonical ensemble, or Gibbs measure, is the probability measure 
$P_{n,\beta}$ on $\Omega_n$ with density (with
respect to $P_n$): 
$$f(\sigma):=
\frac{1}{Z}\exp\left(\frac{\beta}{2n}  \sum_{i,j=1}^n\inprod{\sigma_i}{\sigma_j}
\right)=\frac{1}{Z}e^{-\beta H_n(\sigma)}=\frac{1}{Z}e^{
-n\beta h_n(\sigma)}.$$
Here the partition function is 
$Z = Z_n (\beta) = \int_{\Omega_n} \exp\left( \frac{\beta}{2n} \sum_{i,j=1}^n\inprod{\sigma_i}{\sigma_j} \right) dP_n.$

The empirical measure $\mu_\sigma = \mu_{n,\sigma}$ of the spins $\{\sigma_i\}$ is defined to be the random measure $\mu_{n,\sigma}:=\frac{1}{n}\sum_{i=1}^n\delta_{\sigma_i}$ on $\s^2$. An interesting physical quantity is the total empirical spin, defined by
\[ S_n (\sigma) := n\int xd\mu_\sigma(x)=\sum_{i=1}^n \sigma_i. \]

For a probability
measure $\nu$ on $\s^2$, define the relative entropy of $\nu$ with respect to 
the uniform probability measure $\mu$ by
\[H(\nu\mid\mu):=\begin{cases}\int_{\s^2}f\log(f)d\mu& if \,\, f:=\frac{d\nu}{d\mu} 
\,\,exists;\\\infty&otherwise.\end{cases}\]

Let $M_1(\s^2)$ denote the probability measures on $\s^2$ with the weak-*
topology, and also define:
\[z_n(\beta):=\int_{\Omega_n}e^{-\beta h_n}dP_n\qquad{\rm and}\qquad
p_{n,\beta}(B):=\frac{1}{z_n(\beta)}\int_Be^{-\beta h_n}dP_n.\]
 
A question of significant interest is the behavior of the
total spin as a function of the inverse temperature $\beta$ in the
Gibbs measures, so we begin by stating large deviations
principles (LDP) for the $ \mu_{n,\sigma}$, first a proposition for the noninteracting case $\beta=0$, then a theorem for general $\beta$, followed by an alternative formula for the free energy. 
 The proposition is simply a particular instance of  Sanov's theorem (see Theorem 6.2.10 of \cite{DZ}).
\begin{prop}\label{LDP1} For $P_n$ the $n$-fold product of uniform
  measure on 
$\s^2$ and $\mu_{n,\sigma}=\frac{1}{n}\sum_{i=1}^n\delta_{\sigma_i}$
as above,
if $\Gamma$ a Borel subset of $M_1(\s^2)$,
\[-\inf_{\nu\in\Gamma^\circ}H(\nu\mid\mu)\le\liminf_{n\to\infty}\frac{1}{n}
\log P_n[\mu_{n,\sigma}\in\Gamma]\le\limsup_{n\to\infty}\frac{1}{n}\log P_n[\mu_{n,\sigma}\in\Gamma]\le-\inf_{
\nu\in\overline{\Gamma}}H(\nu\mid\mu);\]
that is, the random measures $\mu_{n,\sigma}$ satisfy an LDP with rate function 
$H(\cdot\mid\mu)$.
\end{prop}

In particular, this  says that the canonical macrostates
$\mathcal{E}_0:=\{\nu:H(\nu\mid\mu)=0\}$ at $\beta=0$ are disordered
since  the rate function $H(\cdot\mid\mu)$ achieves its minimum of $0$
for the uniform measure $\mu$ only. The positive $\beta$ canonical
macrostates are $\mathcal{E}_\beta:= \{\nu: I_\beta(\nu)=0\}$, with
$I_\beta$  defined below.  The following theorem identifies them abstractly; Theorem \ref{freeenergyresults} below describes them concretely.  

\begin{thm}\label{mainLDP}
With notation as above, the $\mu_{n,\sigma}$
satisfy an LDP on $M_1(\s^2)$ with rate function 
\begin{equation}\label{beta_rf}
I_\beta(\nu):=H(\nu\mid\mu)-
\frac{\beta}{2}\left|\int_{\s^2}xd\nu(x)\right|^2-\varphi(\beta),
\end{equation}
where the free energy $\varphi(\beta):=-\lim_{n\to\infty}\frac{1}{n}\log z_n(n\beta)=-\lim_{n\to\infty}
\frac{1}{n}\log Z_n(\beta)$ exists and is given by 
\begin{equation}\label{free-energy}\varphi(\beta)=
\inf_{\nu\in M_1(\s^2)}\left[ H(\nu\mid\mu)-
\frac{\beta}{2}\left|\int_{\s^2}xd\nu(x)\right|^2\right].\end{equation}
Moreover, for every $\beta > 0$, every subsequence of $P_{n,\beta}\left[\mu_{n,\sigma}\in\cdot\right]$
   has a further subsequence converging weakly to $\Pi_\beta$, a probability measure on $M_1(\s^2)$ concentrated on $\mathcal{E}_\beta$, i.e., $\Pi_\beta(\mathcal{E}_\beta^c) = 0.$
In the case that $\mathcal{E}_\beta=\{\nu\}$ for one $\nu\in M_1(\s^2)$, then the whole sequence
converges weakly to $\delta_{\nu}$. 

\end{thm}

The proof of Theorem \ref{mainLDP} makes use of an argument due to
Ellis, Haven and Turkington (see Theorems 2.4 and 2.5 of \cite{EHT}) in order to obtain an
LDP for the empirical spins $\mu_{n,\sigma}$ in the general case of $\beta>0$,
based on the result for the $\beta=0$ (independent) case.  

The analysis in the appendix leads to the following.

\begin{thm}\label{freeenergy}
 The free energy $\varphi$ has the formula (see Lemmas \ref{subcrit} and \ref{supcrit} in the Appendix):
\[ \varphi(\beta) = \begin{cases} 0, \quad \quad \quad \; \text{ if } \beta <3, \\ \Phi_\beta(g^{-1}(\beta)), \text{ if } \beta\ge 3, \end{cases}\] 
where $\Phi_\beta$ is defined by \begin{equation}\label{Phi}
\Phi_\beta(k) := \log\left(\frac{k}{\sinh(k)}\right)+k\coth(k)-1-\frac{\beta}{2}
\left(\coth(k)-\frac{1}{k}\right)^2 
\end{equation}
 and \[ g(x) := \frac{x}{\coth x - 1/x} = \beta.\]

 In particular, the function $\varphi$ and its derivative $\varphi'$ are continuous at the critical threshold $\beta = 3$, so the phase transition is continuous, or second-order.
\end{thm}

\medskip

We remark that the critical value of $\beta=3$ identified above agrees with
the large-dimensional limit of Kesten and Schonmann \cite{KS}.

\medskip

As a corollary of the Sanov theorem for the noninteracting case $\beta=0$, or independently, as a Cram\'er theorem for random vectors on the sphere, one can prove the following Cram\'er-type LDP for the total spin $M_n := \frac1n \sum_{i=1}^n \sigma_i$. 

\begin{cor}\label{spinLDP}Let $\{\sigma_i\}_{i=1}^n$ be i.i.d.\ uniform random points on $\s^2\subseteq\R^3$.
The total spins $M_n$  satisfy an LDP with rate function $I$:
\[P_n \left( M_n\simeq x \right) \simeq e^{-nI(x)},\]
where $I(x) = c|x|-\log\left(\frac{\sinh(c)}{c}\right)=c\coth(c)-1+\log\left(
\frac{c}{\sinh(c)}\right)$
and $c$ is defined by 
\(\coth(c)-\frac{1}{c}=|x|.\)

\end{cor}
\noindent {\bf Remarks:}\begin{enumerate}
\item  The function $g(y)=\coth(y)-\frac{1}{y}$ is strictly increasing on $(0,\infty)$, so that the equation above does uniquely define $c$ as a function of $|x|$.  
	\item Readers familiar with Cram\'er's theorem may expect to
          see the rate function simply identified abstractly  as the
          Legendre-Fenchel transform of the uniform measure on the
          unit sphere.  The explicit rate function $I(x)$ above is
          indeed the Legendre-Fenchel transform of the uniform measure
          on the sphere, although it takes some computation to verify
          this.  A proof of the formula as a consequence of
          Proposition \ref{LDP1} together with further analysis of the
          relative entropy is sketched below, immediately after the
          statement of Theorem \ref{freeenergyresults}. See also the analysis in
          the appendix of the free energy and Equation \eqref{Phi}. 
\end{enumerate}
\medskip

This noninteracting Cram\'er corollary has a companion result for the interacting $\beta>0$ case, which is used in Sections \ref{S:subcrit} -- \ref{S:critical}.
\begin{thm}\label{spinLDPbeta}  Let $P_{n,\beta}$ be the Gibbs measure defined above, and let
  $M_n=M_n(\sigma):=\frac1n \sum_{i=1}^n \sigma_i$.  Then for a Borel set $\Gamma\in\R$,
\[-\inf_{x\in\Gamma^\circ}I_\beta(x)\le\liminf_{n\to\infty}\frac{1}{n}
\log
P_{n,\beta}\left[\beta M_n\in\Gamma\right]\le\limsup_{n\to\infty}\frac{1}{n}\log
P_{n,\beta}\left[\beta M_n\in\Gamma\right]\le-\inf_{
x\in\overline{\Gamma}}I_\beta(x)\]
where 
\[I_\beta(x)=c\coth(c)-1-\log\left(\frac{\sinh(c)}{c}\right)
-\frac{\beta}{2}\left|\coth(c)-\frac{1}{c}\right|^2,\]
and $c$ is the unique element of $\R^+$ such that $|x|=\coth(c)-\frac{1}{c}$.
\end{thm}
One can derive this explicit Cram\'er-type LDP as a special case of Theorem \ref{mainLDP} by considering the center of mass of $\mu_{\sigma,n}$ and using the computations below. Notice also that $I_\beta(x) = \Phi_\beta(c)$ from \eqref{Phi}, where $x$ and $c$ satisfy the above formula.
This result can also be proved directly by standard methods: the case $\beta=0$ is Cram\'er's theorem for random vectors on the sphere with rate function $I=I_{\beta=0}$ (Corollary \ref{spinLDP}), and the general case follows from Theorems 2.4 and 2.5 of \cite{EHT}, similar to the proof of Theorem \ref{mainLDP}.
\medskip

The proof of Theorem \ref{mainLDP}, obtaining the LDP for the empirical measure of the spins at positive $\beta$, proceeds by applying Theorems 2.4 and 2.5 of \cite{EHT} to $z_n(\beta)$ and $p_{n,\beta}$, and identifying the hidden (Polish)
space as $M_1(\s^2)$, the set of Borel probability measures on $\s^2$
equipped with the weak-* topology.  The hidden process is 
$\{\mu_{n,\sigma}\}_{n=1}^\infty$ as above, which satisfies an 
LDP with rate function $H(\cdot\mid\mu)$.  The representation in
question here is of the energy per particle, rather than the Hamiltonian
itself:
\begin{equation*}\begin{split}
h_n(\sigma)&=-\frac{1}{2n^2}\sum_{i,j=1}^n\inprod{\sigma_i}{\sigma_j}=
-\frac{1}{2}\inprod{\int_{\s^2}xd\mu_{n,\sigma}(x)}{\int_{\s^2}xd\mu_{n,\sigma}(x)}=
-\frac{1}{2}\left|\int_{\s^2}xd\mu_{n,\sigma}(x)\right|^2;
\end{split}\end{equation*}
we define $\tilde{h}:M_1(\s^2)\to\R$ by $\tilde{h}(\nu):=-\frac{1}{2}\left|
\int_{\s^2}xd\nu(x)\right|^2$.
Note that the expression inside the norm is simply the center of mass
of the measure $\nu$.

To identify the measures in $\mathcal{E}_\beta$ explicitly, first observe
 that if $f=\frac{d\nu}{d\mu}$ then $H(\nu\mid\mu)=\int f\log(f)d\mu$ depends only
on the value distribution of $f$; that is, (roughly speaking) once the values that $f$ takes on the $\mu$-frequency with which they are taken on are fixed, the first term is determined.  This is quite easy to see if $f$ takes on only finitely many values: suppose that $f(x)=\sum_{i=1}^na_i\1_{A_i}(x)$ with the $a_i$ distinct and the $A_i$ pairwise disjoint.  Then 
\[H(\nu\mid\mu)=\sum_{i=1}^na_i\log(a_i)\mu(A_i),\]
and so $H(\nu\mid\mu)$ depends only on the $a_i$ and the $\mu(A_i)$.  More generally, it follows from Fubini's theorem that 
\[\int f\log(f)d\mu=\int_0^\infty \mu\big[f\log(f)>t\big]dt-\int_0^\infty \mu\big[f\log(f)<-t\big]dt.\]

Once the value distribution of $f$ is fixed, it is then easy to see that the expression
$\left|\int xd\nu(x)\right|$ is maximized for corresponding densities which are 
symmetric about a fixed
pole and decreasing as the distance from the pole increases.  Consider, then,
the case that $f=\frac{d\nu}{d\mu}$, a density that is symmetric about the north pole 
and decreasing away from the pole. 
That is, $\nu_g$ is the measure with density $f(x,y,z)=g(z)$ which is increasing in $z$. Then
\begin{equation*}\begin{split}
H(\nu_g\mid\mu)&=\frac{1}{4\pi}\int_0^{2\pi}\int_0^\pi g(\cos(\theta))\log[
g(\cos(\theta))]\sin(\theta)d\theta d\varphi\\&=\frac{1}{2}\int_{-1}^1g(x)\log[
g(x)]dx.
\end{split}\end{equation*}
By the same substitution,
\[\int_{\s^2} vd\nu_g(v)=\begin{bmatrix}0\\0\\1\end{bmatrix}\int_{-1}^1
\frac{xg(x)}{2}dx.\]
The problem is thus to minimize
\[\frac{1}{2}\int_{-1}^1g(x)\log[
g(x)]dx-\frac{\beta}{2}\left(\int_{-1}^1\frac{xg(x)}{2}dx\right)^2\]
for $g:[-1,1]\to\R_+$ such that $\frac{1}{2}\int_{-1}^1g(x)dx=1$ and
$g$ is increasing.  Observe that 
\[\frac{1}{2}\int_{-1}^1g(x)\log[g(x)]dx=\frac{1}{2}\int_{-1}^1g(x)
\log\left[\frac{g(x)}{2}\right]dx+\log(2)=-h\left(\frac{g}{2}\right)+\log(2),
\] 
where $h(\varphi)$ is
the (usual) entropy of the density $\varphi$.

Now fix the value
of $\left|\int xd\nu(x)\right|\in[0,1]$ and minimize
$\frac{1}{2}\int_{-1}^1g(x)\log[
g(x)]dx$  
over the $\nu\in M_1(\s^2)$ corresponding to this value;
this is a constrained entropy maximization  problem, for which 
known results (see Theorem 12.1.1 from 
\cite{CT}) imply:
\begin{prop}\label{entropy_max}
Consider the class of $f:[-1,1]\to\R_+$ such that
\begin{itemize}
\item $\int_{-1}^1f(x)dx=1$, and 
\item $\left|\int_{-1}^1xf(x)dx\right|=c$.
\end{itemize}
Then $f^*(x)=k_1e^{k_2x}$ uniquely maximizes $h(f)$ over the 
densities satisfying these conditions.
\end{prop}
Now, to determine $k_1,k_2$, observe that for $f^*$ to satisfy the first
condition,
\[1=\int_{-1}^1k_1e^{k_2x}dx=\frac{2k_1\sinh(k_2)}{k_2},\]
and thus 
\[k_1=\frac{k_2}{2\sinh(k_2)}.\]
For the second condition, 
\begin{equation*}\begin{split}
c&=k_1\int_{-1}^1xe^{k_2x}dx=k_1\left[\frac{2\cosh(k_2)}{k_2}-
\frac{2\sinh(k_2)}{k_2^2}\right]=\coth(k_2)-\frac{1}{k_2}.
\end{split}\end{equation*}
Take $g^*=2f^*$;  considering all 
$c\in[0,1]$ and requiring $g^*$ to be increasing corresponds
to considering all $k_2\in[0,\infty)$.
In that case, we need to minimize
\begin{equation}\begin{split}\label{tbm}
\frac{1}{2}\int_{-1}^1&g^*(x)\log[g^*(x)]dx-\frac{\beta}{2}c^2\\&=
\log\left(\frac{k_2}{\sinh(k_2)}\right)+k_2\coth(k_2)-1-\frac{\beta}{2}
\left(\coth(k_2)-\frac{1}{k_2}\right)^2 =: \Phi_\beta(k_2)
\end{split}\end{equation}
over all $k_2\in[0,\infty)$.  The problem has thus been reduced to
a one-dimensional calculus exercise, all of whose details are carried
out in Section 7.1 of the Appendix.  
Those calculations lead to a critical value of the inverse temperature
$\beta_c = 3$, and to the fact that the phase transition is a continuous one (2nd order
in physics parlance).  Below the transition, the only macrocanonical
state is the uniform distribution, and then increasing $\beta$ across
the critical threshold, a spherically symmetric family of
distributions with a preferred direction appears.  At first the
direction is hardly preferred at all, but with increasing $\beta$ the
preferred direction becomes more strongly preferred, so that in the
zero temperature limit $\beta \to \infty$, the macrostates are point
masses.  More precisely, we have the following.

\begin{thm}\label{freeenergyresults}

\begin{enumerate}
\item In the subcritical case, $\beta\le 3$, the expression \eqref{tbm} is minimized for $k_2=0$, and the corresponding $k_1=0$, so that the minimizing function $f^* = 1$ and hence the canonical macrostates in the subcritical case are uniform:  $\mathcal{E}_\beta=\{\mu\}$.
\item In the supercritical case, $\beta>3$, the minimizing $k_2$ for the expression \eqref{tbm} is
  the unique strictly positive solution to 
\[x=\beta\left(\coth(x)-\frac{1}{x}\right),\]which moreover
has limit $\lim_{\beta \downarrow \beta_c} k_2=0$.

The macrostates $\mathcal{E}_\beta$ are given by
$\{\nu_x\}_{x\in\s^2},$ where $\nu_x$ is the probability measure with density which is
symmetric about the pole at $x$, with density $g_x:[-1,1]\to\R$ in the
$x$-direction given by $2k_1e^{k_2x}$ with $k_2$ as above.
\end{enumerate}
\end{thm}


The general result on constrained entropy maximization used above also gives a proof of Corollary \ref{spinLDP} from Proposition \ref{LDP1} as follows.

\begin{proof}[Proof of Corollary \ref{spinLDP}]
Given $x_o\in\R^3$ and $\epsilon>0$, take the set $\Gamma$ in Proposition \ref{LDP1} to be $\left\{\nu\in M_1(\s^2):\left|\int_{\s^2}xd\nu(x)-x_o\right|<\epsilon\right\}$.  One must then consider
\[\inf\left\{\int f\log(f)d\mu :\left|\int_{\s^2}xf(x)d\mu(x)-x_o\right|<\epsilon\right\}.\]
This is exactly the constrained entropy maximization problem addressed in Proposition \ref{entropy_max} and the analysis which followed, from which the form of the rate function stated in Corollary \ref{spinLDP} follows.
\end{proof}

We conclude this section by giving a treatment of the microcanonical ensemble, in which one fixes the
energy per particle.  The following result gives an LDP in that
case using the results of \cite{EHT}.
\begin{prop} Since in the i.i.d.\ case, the empirical measure of the spins $\mu_{n,\sigma}
  = \frac1n \sum \delta_{\sigma_i}$ satisfies an LDP 
with rate $H(\cdot \mid \mu)$, we have the following:
\begin{enumerate}
\item The energies $\tilde{h}(\mu_{n,\sigma})$ and $H_n$ satisfy LDPs with rate $J$, called the microcanonical entropy, defined for a fixed value $u$ of the energy by:
\[ J(u) := \inf \{H(\nu\mid\mu) : \nu \in M_1(\s^2), \; \tilde{h}(\nu) = u \}. \]
The free energy $\varphi$ is a Legendre-Fenchel transform: $\varphi(\beta) = \inf_{u \in \R} \{ \beta u + J(u) \}. $ 
\item If $u \in dom(J)$, define the microcanonical Gibbs measure  by
\[ P^{u,r}_n (A) := \frac{1}{Z_u} \int_A \1_{\{H_n \in [u+r, u-r]\}} dP_n,\]
where $Z_u := \int_{\s^2} \1_{\{H_n \in [u+r, u-r]\}} dP_n$.  Then for
$\sigma$  distributed according to $P_n^{u,r}$, $\mu_{n,\sigma}$
satisfies an LDP 
with microcanonical rate function
\[ I^u (\nu) := \begin{cases} H(\nu\mid\mu) - J(u),& \text{if } -\frac12 \inprod{\nu}{\nu} = u;\\\infty,& \text{otherwise.}\end{cases}\]
That is,
\[-\inf_{\nu\in\Gamma^\circ}I^u(\nu)\le\lim_{r \to 0}\liminf_{n\to\infty}\frac{1}{n}
\log P^{u,r}_n(\mu_{n,\sigma}\in\Gamma)\le\lim_{r \to 0}\limsup_{n\to\infty}\frac{1}{n}\log P^{u,r}_n(\mu_{n,\sigma}\in\Gamma)\le-\inf_{\nu\in\overline{\Gamma}}I^u(\nu);\]
\item The microcanonical macrostates are 
\[ \mathcal{E}^u := \{ \nu: H(\nu\mid\mu) = J(u), \; \tilde{h}(\nu) = u \}.\]
\end{enumerate}
\end{prop}
Again, it suffices to restrict our attention to symmetric densities $\frac{d\nu}{d\mu}$, symmetric about a pole with unit vector $\hat{z}$, i.e., $\int vd\nu(v)=c\hat{z}$, because symmetrizing about the $z$-axis reduces relative entropy: Recall
that $d\mu=d\mu_z\frac{dz}{2}$, where $\mu_z$ is the uniform measure
on the circle of radius $\sqrt{1-z^2}$.  Let 
$\tilde{f}(z):=\int f(x,y,z)d\mu_z(x,y)$, which is the symmetrized version 
of $f$ about the $z$-axis.  Now, the function $g(x)=x\log(x)$
is convex, so by Jensen's inequality,
\begin{equation*}\begin{split}
\tilde{f}(z)\log[\tilde{f}(z)]=g\left(\int f(x,y,z)d\mu_z(x,y)\right)&\le
\int g(f(x,y,z))d\mu_z(x,y)\\&=\int f(x,y,z)\log[(f(x,y,z)]d\mu_z(x,y).
\end{split}\end{equation*}
Integrating both sides with respect to $\frac{dz}{2}$ shows that
$H(\tilde{\nu}\mid\mu)\le H(\nu\mid\mu)$, where  $\nu$ and $\tilde{\nu}$ are respectively the 
measures with densities $f$ and $\tilde{f}$.  

We can compute $J$ to be
\[ J(u) = \inf \left\{ H(\tilde{\nu} \mid \mu) : \left( \frac12 \int_{-1}^1 x \tilde{f}(x) dx \right)^2 = -2u, \: \frac12 \int \tilde{f}(x)dx = 1\right\}, \]
and then simplify it using the previous result on maximizing entropy, with $k_2$ solving $\coth k_2 - \frac{1}{k_2} = \sqrt{-2u}$:
\[ J(u) = \log \left(\frac{k_2}{\sinh k_2}\right) + k_2 \sqrt{-2u}.\]
The microcanonical entropy is
\[ I^u(\tilde{\nu})  = -h(\tilde{f}/2) + \log 2 - J(u), \text{ if } \left( \frac12 \int_{-1}^1 x \tilde{f}(x) dx \right)^2 = -2u. \]

The domain of $J$ is $\left(-\frac{1}{2},0\right]$, and the microcanonical macrostates $\mathcal{E}^u$ consist of rotations to any direction of $\tilde{\nu}$ with density $$\tilde{f} = \frac{k_2}{\sinh k_2} e^{k_2 x}, \text{ where } \coth k_2 - \frac{1}{k_2} = \sqrt{-2u}.$$ In particular, $\mathcal{E}^0 = \{\mu\}$, the completely disordered phase, and for energies close to zero $-\frac12 \ll u < 0$, there is the expansion $J(u) \simeq -3u -\frac92 u^2 + \dots$, and $k_2 \simeq 3\sqrt{-2u}$. Thus the microcanonical macrostates for small energy are $\mathcal{E}^{u} = \{\nu_x\}_{x \in \s^2}$, where $\nu_x$ is the rotation of $\tilde{\nu}$ to the $x$ direction, again a continuous transition to the ordered phase. 


\section{Limit theorems for the total spin}
In each regime (subcritical, critical, supercritical), the total spin
satisfies a limit theorem.  For convenience, we collect these results
here; the proofs are in the subsequent three sections.

In the subcritical regime, there is the following multivariate central
limit theorem.
\begin{thm} \label{subcrit_limit}For $\beta<3$, there is a constant $c_\beta$ depending
  only on $\beta$ such that for $W_n=\sqrt{\frac{3-\beta}{n}}
\sum_{i=1}^n\sigma_i$,  
\[\sup_{g:M_1(g),M_2(g)\le 1}|\E g(W_n)-\E g(Z)|\le\frac{c_\beta\log(n)}{\sqrt{n}}\]
where $M_1(g)$ is the Lipschitz constant of $g$, $M_2(g)$ is the
maximum operator norm of the Hessian of $g$, and $Z$ is a standard
Gaussian random vector in $\R^3$.
\end{thm}

The form of the theorem above may be slightly unfamiliar to some
readers, so it seems worth noting explicitly that for random vectors
$X$ and $Y$ in $\R^d$, the quantity
\[\sup_{g:M_1(g),M_2(g)\le 1}\big|\E g(X)-\E g(Y)\big|\] is a metric for
the familiar topology of weak-star convergence together with convergence in mean on the
space of probability measures.   We have stated the result in the form
above because the rate of convergence is probably sharp, up to the
logarithmic factor.  However, if one prefers the more usual
$L_1$-Wasserstein distance as a metric for this topology, the analysis
in Section \ref{S:subcrit} yields the following rate of convergence for a multivariate
limit theorem there.  The $L_1$-Wasserstein distance has several
equivalent definitions; the one most relevant to us is the following:
let $X$ and $Y$ be random vectors in $\R^n$.  Then the
$L_1$-Wasserstein distance $d_W(X,Y)$ between $X$ and $Y$ is defined
by
\[d_W(X,Y)=\sup_{g:M_1(g)\le 1}\big|\E g(X)-\E g(Y)\big|,\]
where as above, $M_1(g)$ denotes the Lipschitz constant of $g$. 
\begin{thm}\label{wasserstein}
For $W_n$ constructed as above and $Z$ a standard Gaussian random
vector in $\R^3$,
\[d_W(W_n,Z)\le\frac{c_\beta}{n^{1/4}},\]
where $c_\beta$ is a constant depending only on $\beta$.
\end{thm}

\medskip

In the ordered regime, where $\beta$ is large,  the spins
  tend to align.  Indeed, it follows from the large deviations
  principle for $S_n$ that $\left|S_n\right|$ is close to
  $\frac{k_2n}{\beta}$ with high probability: apply Theorem \ref{spinLDPbeta} with $\Gamma$ a small interval around
  $k_2$, and use the fact that $k_2$ is the
  argmin of $I_\beta$.  It is also true that $S_n$ is {\em a priori}
  spherically symmetric, making any limiting point on the sphere of
  radius $\frac{k_2n}{\beta}$ equally likely.  This makes the limiting situation in the
  ordered regime (and at criticality, discussed below) quite different
  from that of the disordered regime described above; rather than a
  limiting distribution for $S_n$ about one deterministic point (i.e., zero),
  one must consider the fluctuations of $S_n$ about a spherically symmetric
  family of possible limiting values.  In the context of statistical mechanical
  models of this type (i.e., the Curie-Weiss or Curie-Weiss-Potts
  models), this situation has typically been treated by conditioning
  on the limiting direction of the total spin, and then considering
  the conditional fluctuations about that limit (see,  e.g.,  \cite{ENR}).
Here we address this issue by treating instead the fluctuations of the squared-length of the
total spin; that is, we consider the random variable
  \begin{equation}\label{W-def}
W_n:=\sqrt{n}\left[\frac{\beta^2}{n^2k_2^2}\left|\sum_{j=1}^n\sigma_j
\right|^2-1\right].\end{equation}
 In Section \ref{S:supcrit} it is shown that $W_n$ satisfies the central limit
  theorem stated below.  Since the distribution of the total spin
  $\frac{1}{n}\sum_{j=1}^n\sigma_j$ is rotationally invariant, this gives a
  complete picture of its asymptotic behavior without making use of conditioning.

\smallskip

For technical reasons, the following result is given in terms of the
so-called bounded-Lipschitz distance between $W_n$ and $Z$ rather than
in the Wasserstein distance; bounded-Lipschitz distance
is a metric for the topology of weak-star convergence of probability
measures. The bounded-Lipschitz distance $d_{BL}(X,Y)$ between random
variables $X$ and $Y$ is defined by  
\begin{equation}\label{bldist}d_{BL}(X,Y):=\sup\left\{\Big|\E h(X)-\E
    h(Y)\Big|:\|h\|_\infty\le 1, M_1(h)\le 1\right\},\end{equation}
where $\|h\|_\infty$ is the supremum norm of $h$ and $M_1(h)$ is again
the
Lipschitz constant of $h$.  We note that the definition of
bounded-Lipschitz distance is sometimes given in terms of probability
measures than random variables, but the two viewpoints are of course
completely equivalent since the definition above depends only on the
distributions of $X$ and $Y$.

\begin{thm}\label{T:supcrit_CLT}
Let $W_n$ be the
recentered, renormalized norm squared of the total spin, as defined
in \eqref{W-def}.  There is a constant $c_\beta$ depending only on
$\beta>3$ such that
if $Z$ is a centered Gaussian random variable with variance 
\[\sigma^2:=\frac{4\beta^2}{\left(1-\beta g'(k_2)\right)k_2^2}
\left[\frac{1}{k_2^2}-\frac{1}{\sinh^2(k_2)}
\right],\]
for $g(x)=\coth(x)-\frac{1}{x}$, then

\[d_{BL}(W_n,Z)\le c_\beta\left(\frac{\log(n)}{n}\right)^{1/4}.\]
\end{thm}

\medskip

In Section \ref{S:critical}, we prove the following  nonnormal limit theorem for the
squared-length of the total
total spin at the critical temperature $\beta=3$.  Again, since the
total spin is spherically symmetric, this provides the limiting
picture in the critical case.  

\begin{thm}\label{T:limit_crit}
At the critical temperature $\beta=3$, let
$W_n:=\frac{c_3|S_n|^2}{n^{3/2}}$, where $c_3$ is such that $\E W_n=1$.
Let $X$ have density
\[p(t)=\begin{cases}\frac{1}{z}t^5e^{-3ct^2}&t\ge
  0;\\0&t<0,\end{cases}\]
where $c=\frac{1}{5c_3}$ and $z$ is a normalizing factor.
Then there is a universal constant $C$ such that 
\[\sup_{\substack{\|h\|_\infty\le 1, \,\|h'\|_\infty\le 1\\\|h''\|_\infty\le
1}}\big|\E h(W_n)-\E h(X)\big|\le\frac{C\log(n)}{\sqrt{n}}.\]
\end{thm}

\noindent {\em Note:} the quantity 
\[\sup_{\substack{\|h\|_\infty\le 1, \,\|h'\|_\infty\le 1\\\|h''\|_\infty\le
1}}\big|\E h(Y)-\E h(X)\big|\]
is a metric for the weak-star topology on random variables, so that
the result above should be viewed as a limit theorem with an explicit
rate of convergence in this metric.  As in the
subcritical case, one can employ a standard smoothing argument to
obtain a rate of convergence in a more familiar metric, in this case,
the bounded-Lipschitz distance defined above in Equation \eqref{bldist}.
\begin{thm}\label{T:bldist}
For $W_n$ and $X$ as above, there is a
universal constant $C$ such that
\[d_{BL}(W_n,X)\le\frac{C\log(n)}{n^{3/8}}.\]
\end{thm}

\noindent {\em Remark:} The reader may have noted that the limit
theorem in the subcritical case can be formulated  in the
$L_1$-Wasserstein distance, whereas those in the critical and
supercritical cases are in the bounded-Lipschitz distance, which
metrizes a slightly weaker topology.  This is a typical by-product of
the technical differences between Stein's method in multivariate (as
in the subcritical case) and univariate (as in the critical and
supercritical cases) settings and does not reflect essential differences.
\section{The total spin in the subcritical phase}\label{S:subcrit}

In this section we give proofs of Theorems \ref{subcrit_limit} and
\ref{wasserstein}, giving the limit theorem for $S_n$ in the
disordered regime.

While it is not formally necessary, we find it helpful to give a
heuristic computation of the variance of the total spin
$S_n:=\sum_{i=1}^n\sigma_i$ before proceeding with rigorous proofs.
Note that each of the spins $\sigma_i$ has a uniform marginal
distribution, because the density of the Gibbs measure is rotationally
invariant.  
Also,  $\E\inprod{\sigma_i}{
\sigma_i}=1$ for each $i$;  moreover, by the symmetry of the Gibbs 
measure, $\E\inprod{\sigma_i}{\sigma_j}$ is the same for every pair $i\neq j$.
Now, conditional on $\{\sigma_j\}_{j\neq 1}$, 
the density of $\sigma_1$ with respect to uniform measure on $\s^2$ is 
given by 
\[\frac{1}{Z_1}\exp\left[\frac{\beta}{n}\sum_{j\neq 1}\inprod{\theta}{\sigma_j}
\right],\]
where the normalization is $Z_1=\int_{\s^2}\exp\left[\frac{\beta}{n}\sum_{j\neq 1}\inprod{\theta}{\sigma_j}
\right]d\mu(\theta).$
For fixed $i\in\{1,\ldots,n\}$, let $\sigma^{(i)}:=\sum_{j\neq i}\sigma_j$.
Note that $d \mu(\theta) = \frac{\sin \alpha}{4\pi}
 d \alpha d \phi$, where $(\alpha, \phi)$ are spherical coordinates; $Z_1$ is therefore given by
\[Z_1=\frac{1}{4\pi} \int_0^{2\pi}\int_0^\pi e^{c\cos(\alpha)}\sin(\alpha)d\alpha
d\phi=\frac{\sinh(c)}{c},\] where $c=\frac{\beta|\sigma^{(1)}|}{n}$.  Now, from the 
conditional density above,
\begin{equation*}\begin{split}\E\left[\sigma_1\big|\{\sigma_j\}_{j\neq 1}\right]&=
\frac{1}{Z_1}
\int_{\s^2}\theta\exp\left[\frac{\beta}{n}\sum_{j\neq 1}\inprod{\theta}{\sigma_j}
\right]d\mu(\theta)\\&=\frac{1}{Z_1}
\int_{\s^2}\inprod{\theta}{\frac{
\sigma^{(1)}}{|\sigma^{(1)}|}}\left(\frac{\sigma^{(1)}}{|\sigma^{(1)}|}\right)
\exp\left[\frac{\beta}{n}\inprod{\theta}{\sigma^{(1)}}\right]d\mu(\theta)\\&=
\left[\frac{1}{4\pi Z_1}\int_0^{2\pi}\int_0^\pi
\cos(\alpha)\sin(\alpha)e^{\frac{\beta|\sigma^{(1)}|\cos(\alpha)}{n}}d\alpha d\phi
\right]\left(
\frac{\sigma^{(1)}}{|\sigma^{(1)}|}\right)\\
&=\left[\coth\left(\frac{\beta|\sigma^{(1)}|}{n}
\right)-\frac{n}{\beta|\sigma^{(1)}|}
\right]\left(\frac{\sigma^{(1)}}{|\sigma^{(1)}|}\right).
\end{split}\end{equation*}
It follows from Theorem \ref{spinLDPbeta} that if $\beta<3$, then $\frac{\beta|\sigma^{(1)}|}{n}=o(1)$ with 
probability exponentially close to one: take $\Gamma=[a,b]$ for any
$a>0$, and recall (see Theorem \ref{freeenergyresults}) that if
$\beta<3,$ then $I_\beta(x)$ has its minimum value of zero only at $x=0$.  Expanding about zero gives that for
$x$ small, $\coth(x)-\frac{1}{x}\approx\frac{x}{3}$, and so if 
$\beta<3$, then 
\[\E\left[\sigma_1\big|\{\sigma_j\}_{j\neq 1}\right]\approx
\frac{\beta\sigma^{(1)}}{3n}=\frac{\beta}{3n}\sum_{i\neq 1}\sigma_i.\]
It follows by taking inner product of both sides with $\sigma_2$ followed
by expectation that
\begin{equation*}\begin{split}
\E\inprod{\sigma_1}{\sigma_2}&\approx
\frac{\beta}{3n}\E\inprod{\sum_{j\neq1}\sigma_j}{
\sigma_2}=\frac{\beta}{3n}\left[1+(n-2)\E\inprod{\sigma_1}{
\sigma_2}\right],
\end{split}\end{equation*}
and thus
\[\E\inprod{\sigma_1}{\sigma_2}\approx\frac{\beta}{3n-\beta(n-2)}\approx
\frac{\beta}{n(3-\beta)}.\]
Finally,
\[\E |S_n|^2=n\E|\sigma_1|^2+n(n-1)\E\inprod{\sigma_1}{\sigma_2}
\approx\frac{3n}{3-\beta}.\]

\medskip
Theorem \ref{subcrit_limit} is proved as an application of the following abstract
normal approximation theorem from \cite{Me}.  
\begin{thm}\label{normal_approx}
Let $(X,X')$ be an exchangeable pair of random vectors in $\R^d$.  Let         
$\F$ be a $\sigma$-algebra with $\sigma(X)\subseteq\F$, and suppose that       
there is an invertible matrix $\Lambda$, a symmetric, positive definite    
matrix $\Sigma$, an $\F$-measureable random vector $R$ and an
$\F$-measureable random matrix $R'$ such that                                                   
\begin{enumerate}                                                              
\item \label{lincond}                                                          
$$\E\left[X'-X\big|\F\right]=-\Lambda X+R$$                                    
\item \label{quadcond}                                                         
$$\E\left[(X'-X)(X'-X)^T\big|\F\right]=2\Lambda\Sigma+R'.$$                    
\end{enumerate}                                                                
Then for $g\in C^2(\R^d)$,                        
\begin{equation}\begin{split}\label{bd2}                                       
\big|\E g(X)-\E g(\Sigma^{1/2}Z)\big|\le M_1(g)&\|\Lambda^{-1}\|_{op}\left[    
\E|R|+\frac{1}{2}\|\Sigma^{-1/2}\|_{op}\E\|R'\|_{H.S.}                         
\right]\\                                                                      
&+\frac{\sqrt{2\pi}}{24}M_2(g)\|\Sigma^{-1/2}\|_{op}\|\Lambda^{-1}\|_{op}  
\E|X'-X|^3,\end{split}\end{equation}      
where $M_1(g)$ is the Lipschitz constant of $g$ and $M_2(g)$ is the maximum
operator norm of ${\mathrm {Hess}}(g)$.
\end{thm}

Theorem \ref{normal_approx} is a version of Stein's method of exchangeable pairs, introduced and
developed in Stein's book \cite{St}, and subsequently built upon by
many researchers in many contexts.  One of the great virtues of the
method is that it does produce limit theorems with 
explicit error bounds, as we have indicated above. 

The theorem itself may appear rather abstract and unmotivated; it is
not obvious why such conditions should lead to Gaussian behavior.
Considering the univariate case for simplicity, it
may be helpful to note that if $(X,X')$ were jointly Gaussian random
variables and
also close, then they would be of the form
$(Z_1,\sqrt{1-\epsilon^2}Z_1+\epsilon Z_2)$ for $Z_1$ and $Z_2$
independent Gaussian random variables and $\epsilon$ small.  A quick computation shows
that conditions $(a)$ and $(b)$ would indeed hold in that case with
$R=0$ and $R'$ of order $\epsilon^4$.  For further background on
Stein's method, see \cite{BaCh}.

\medskip

In order to apply Theorem \ref{normal_approx}, one must construct an exchangeable pair $(W_n,W_n')$.
As is frequently the case in this type of argument, the exchangeable pair will
first be constructed on the level of configurations, and then descend
to the total spin.  Given a fixed configuration $\sigma$, construct
a new configuration $\sigma'$ by letting $I$ 
be distributed uniformly in $\{1,\ldots,n\}$ and
replacing $\sigma_I$ by $\sigma_I'$, distributed according
to the conditional distribution of the $I$-th spin, given $\{\sigma_j:j\neq
I\}$, and defining $\sigma'_j:=\sigma_j$ for $j\neq I$. This procedure is called the Gibbs sampler.  Then the total spin of the original configuration is 
$W_n=\sqrt{\frac{3-\beta}{n}}\sum_{i=1}^n\sigma_i$ and the total spin of the new configuration is $W_n'=W_n(\sigma')=
W_n-\sqrt{\frac{3-\beta}{n}}
\sigma_I+\sqrt{\frac{3-\beta}{n}}\sigma_I'$. The following lemma gives expressions for $R,R',\Sigma,\Lambda$ in the present context.

\begin{lemma}\label{L:errors}For the exchangeable pair $(W_n,W_n')$ as
  constructed above and $\Lambda=\left(\frac{1-\frac{\beta}{3}}{n}\right)Id,$
\begin{enumerate}
\item \[\E\left[W_n'-W_n\big|\sigma\right]=-\Lambda W_n+R,\]
where 
\[R=-\frac{\beta}{3n^2}W_n+\frac{a}{n^{3/2}}\sum_{i=1}^n\left[\coth\left(\frac{\beta
|\sigma^{(i)}|}{n}\right)-\frac{n}{\beta|\sigma^{(i)}|}-\frac{\beta|\sigma^{
(i)}|}{3n}
\right]\left(\frac{\sigma^{(i)}}{|\sigma^{(i)}|}\right)
;\]
\item  \[\E\left[(W_n'-W_n)(W_n'-W_n)^T\big|\sigma\right]=2\Lambda R',\]
 with
\begin{equation*}\begin{split}
R'&=\left(\frac{1-\frac{\beta}{3}}{n}\right)\left[\frac{1}{n}\sum_{i=1}^n
3\sigma_i\sigma_i^T-Id\right]-\frac{2\beta}{3n^2} W_nW_n^T+\frac{2\beta}{3n^4}\sum_{i=1}^n\sigma_i\sigma_i^T\\&\qquad\qquad+\frac{a^2}{n^2}\sum_{i=1}^n
\left\{\left[\frac{2}{3}-\frac{2c\coth\left(c\right)-2}{c^2}
\right]P_i+\left[\frac{\coth(c)}{c}-\frac{1}{c^2}-\frac{1}{3}\right]
P_i^\perp\right.\\&\qquad\qquad\qquad\qquad\qquad\qquad\qquad\qquad\left.-
\left[\coth(c)-\frac{1}{c}-\frac{c}{3}\right](r_i\sigma_i^T+\sigma_ir_i^T)
\right\}.\end{split}\end{equation*}
\end{enumerate}
In particular, the matrix $\Sigma$ of
Theorem \ref{normal_approx} is simply the identity.
\end{lemma}

The next lemma gives bounds for the quantities $R$ and $R'$ identified
above, from which Theorem \ref{subcrit_limit} follows.

\begin{lemma}\label{L:error_bounds}
For $(W_n,W_n')$ as constructed above and $R,R'$ as in the previous
lemma, there is a constant $c_\beta$ depending only on
  $\beta$, such that
\begin{enumerate} 
\item \(\E|R|\le\frac{c_\beta\log(n)}{n^{3/2}};\)
\item \(\E\|R'\|_{H.S.}\le\frac{c_\beta}{n^{3/2}};\)
\item \(\E|W_n'-W_n|^3\le\frac{c_\beta}{n^{3/2}}.\)
\end{enumerate}
\end{lemma}

Theorem \ref{subcrit_limit} now follows immediately from Theorem
\ref{normal_approx} and Lemmas \ref{L:errors} and
\ref{L:error_bounds}.  Theorem \ref{wasserstein} follows from a
standard smoothing argument, in which one takes a function which is
assumed only to be Lipschitz with Lipschitz constant 1 and convolves
with a centered Gaussian density of variance $\frac{1}{t^2}$, then
optimizes over $t$.  Such an argument is carried out carefully in
Section 3 (see in particular Corollary 3.5) of \cite{MM}.
\begin{proof}[Proof of Lemma \ref{L:errors}]
For notational convenience, 
let $a:=\sqrt{3-\beta}$.

For part (a), using the computation at the beginning of the section
one has
\begin{equation*}\begin{split}
\E\left[W_n'-W_n\big|\sigma\right]&=-\frac{a}{n^{3/2}}\sum_{i=1}^n\left[
\sigma_i-\E\left[\sigma_i\big|\{\sigma_j\}_{j\neq i}\right]\right]
\\&=-\frac{1}{n}W_n+\frac{a}{n^{3/2}}\sum_{i=1}^n\left[\coth\left(\frac{\beta
|\sigma^{(i)}|}{n}\right)-\frac{n}{\beta|\sigma^{(i)}|}
\right]\left(\frac{\sigma^{(i)}}{|\sigma^{(i)}|}\right).
\end{split}\end{equation*}

Now, since $\beta<3$, it is known that $\frac{\beta|\sigma^{(i)}|}{n}=o(1)$
with probability exponentially close to 1.  We therefore use the 
expansion of $\coth(x)-\frac{1}{x}$ near zero to write  
\begin{equation*}\begin{split}
\E\left[W_n'-W_n\big|\sigma\right]&=-\frac{1}{n}W_n+
\left(\frac{a}{3n^{5/2}}\sum_{i=1}^n\sum_{j\neq i}
\beta\sigma_j\right)\\&\qquad\qquad\qquad
+\frac{a}{n^{3/2}}\sum_{i=1}^n\left[\coth\left(\frac{\beta
|\sigma^{(i)}|}{n}\right)-\frac{n}{\beta|\sigma^{(i)}|}-\frac{\beta|\sigma^{
(i)}|}{3n}
\right]\left(\frac{\sigma^{(i)}}{|\sigma^{(i)}|}\right).
\end{split}\end{equation*}
Note that
\[\frac{a}{3n^{5/2}}\sum_{i=1}^n\sum_{j\neq i}
\beta\sigma_j=\frac{1}{3n^2}\sum_{i=1}^n\left(\beta W_n-\frac{a\beta
\sigma_i}{\sqrt{n}}\right)=
\frac{1}{n}\left(\frac{\beta}{3}-\frac{\beta}{3n}\right)W_n.\]

The matrix $\Lambda$ of Theorem \ref{normal_approx} is thus
$\frac{1-\frac{\beta}{3}}{n}Id$ and 
\[R=-\frac{\beta}{3n^2}W_n+\frac{a}{n^{3/2}}\sum_{i=1}^n\left[\coth\left(\frac{\beta
|\sigma^{(i)}|}{n}\right)-\frac{n}{\beta|\sigma^{(i)}|}-\frac{\beta|\sigma^{
(i)}|}{3n}
\right]\left(\frac{\sigma^{(i)}}{|\sigma^{(i)}|}\right)
.\]

We now proceed to the proof of part (b). 
By the same considerations as above, 
\begin{equation*}\begin{split}
\E&\left[(W_n'-W_n)(W_n'-W_n)^T\big|\sigma\right]\\&\qquad=\frac{a^2}{n^2}\sum_{i=1}^n
\frac{1}{Z_i}\int_{\s^2}(\theta-
\sigma_i)(\theta-\sigma_i)^T\exp\left[\frac{\beta}{n}
\sum_{j\neq i}\inprod{\sigma_j}{
\theta}\right]d\mu(\theta)\\&\qquad=\frac{a^2}{n^2}\sum_{i=1}^n
\frac{1}{Z_i}\int_{\s^2}\left[\theta\theta^T-\sigma_i\theta^T-\theta\sigma_i^T
+\sigma_i\sigma_i^T\right]\exp\left[\frac{\beta}{n}\sum_{j\neq i}
\inprod{\sigma_j}{\theta}\right]d\mu(\theta).
\end{split}\end{equation*}
Letting $\theta=\theta_1+\theta_2$, where $\theta_1$ is the projection 
of $\theta$ onto the direction of $\sigma^{(i)}$ and $\theta_2$
is the orthogonal complement, the first term of the $i^{th}$ summand is 
\begin{equation}\label{theta_part}
\frac{1}{Z_i}\int_{\s^2}\big[\theta\theta^T\big]
\exp\left[\frac{\beta}{n}\inprod{\sigma^{(i)}}{
\theta}\right]d\mu(\theta)=\frac{1}{Z_i}\int_{\s^2}\big[\theta_1\theta_1^T+\theta_2\theta_2^T\big]
\exp\left[\frac{\beta}{n}\inprod{\sigma^{(i)}}{
\theta}\right]d\mu(\theta),\end{equation}
since the cross terms vanish by symmetry. 
To compute it,
write $r_i = \frac{
\sigma^{(i)}}{|\sigma^{(i)}|}$, so that $\theta_1 =
  \inprod{\theta}{r_i} r_i,$ and $ \theta_1\theta_1^T = \left|
  \inprod{\theta}{r_i} \right|^2r_ir_i^T;$ setting $c :=
  \frac{\beta|\sigma^{(i)}|}{n},$ 
\[\frac{1}{Z_i}\int_{\s^2} \theta_1\theta_1^T  \exp\left[\frac{\beta}{n}
\inprod{\sigma^{(i)}}{\theta}\right]d\mu(\theta) = \frac{1}{Z_i}
\left( \int_{\s^2} \left|\inprod{\theta}{r_i} \right|^2 \exp\left[ c
  \inprod{r_i}{\theta}\right]d\mu(\theta)\right)r_ir_i^T.\] 
(Recall that $r_ir_i^T$ is orthogonal projection onto the span of $r_i$ in
$\R^3$.)
In spherical coordinates (with $r_i$ playing the role of the north pole), 
the integral is then given by
\[ \frac{1}{4\pi Z_i}\int_0^{2\pi} \int_0^\pi \left| \cos \alpha \right|^2  
\exp\left[ c \, \cos \alpha \right] \sin \alpha d \alpha d\phi.\]
Evaluating and using the established formula for 
 $Z_i$ yields
\[ \frac{1}{Z_i}\int_{\s^2} \theta_1\theta_1^T  \exp\left[\frac{1}{n}
\inprod{\sigma^{(i)}}{\theta}\right]d\mu(\theta) =
\frac{ c^2-2c\coth(c)+2}{c^2}P_i,\]
where $P_i$ is orthogonal projection onto $r_i$.

Now, for the second half of \eqref{theta_part}, let $(\theta_x,\theta_y)$
be a representation of $\theta_2$ in orthonormal coordinates within
$r_1^\perp$.
Note that 
\[\int_{\s^2}\theta_{x}\theta_{y}e^{c\inprod{r_i}{\theta}}d\mu(\theta)=0\]
by symmetry.  Expanding in polar coordinates,
\begin{equation*}\begin{split}
\int_{\s^2}\theta_{y}^2e^{c\inprod{r_i}{\theta}}d\mu(\theta)&=\frac{1}{4\pi}
\int_0^{2\pi}\int_0^\pi\sin(\alpha)^3\cos(\phi)^2e^{c\cos(\alpha)}d\alpha d\phi\\&=
\frac{1}{4}\int_0^\pi\sin(\alpha)^3e^{c\cos(\alpha)}d\alpha\\&=
\frac{c\cosh(c)-\sinh(c)}{c^3}.
\end{split}\end{equation*}
Of course, the value is the same when $\theta_{y}^2$ is replaced by 
$\theta_{x}^2$, and thus
\begin{equation*}\begin{split}
\frac{1}{Z_i}\int_{\s^2}\theta_2\theta_2^T\exp\left[\frac{\beta}{n}
\inprod{\sigma^{(i)}}{\theta}\right]d\mu(\theta)&=\frac{c\cosh(c)-
\sinh(c)}{Z_ic^3}P_i^\perp\\&=\left[\frac{\coth(c)}{c}-\frac{1}{c^2}
\right]P_i^\perp,\end{split}\end{equation*}
where $P_i^\perp$ is the orthogonal projection onto $r_1^\perp$.

Formulae for the middle terms follow from the computations above:
\begin{equation*}\begin{split}
\frac{1}{Z_i}\int_{\s^2}\theta\sigma_i^T\exp\left[\frac{\beta}{n}\inprod{\sigma^{(i)}}{
\theta}\right]d\mu(\theta)&=\left[\coth(c)-\frac{1}{c}\right]r_i\sigma_i^T.
\end{split}\end{equation*}
The next term is just the transpose of this one.

Finally, the last term is trivial:
\[\frac{1}{Z_i}\int_{\s^2}\sigma_i\sigma_i^T\exp\left[\frac{\beta}{n}
\inprod{\sigma^{(i)}}{\theta}\right]d\mu(\theta)
=\sigma_i\sigma_i^T.\]
Collecting terms, 
\begin{equation*}\begin{split}
\E\left[(W_n'-W_n)(W_n'-W_n)^T\big|\sigma\right]&=\frac{a^2}{n^2}\sum_{i=1}^n
\left\{\left[1-\frac{2c\coth\left(c\right)-2}{c^2}
\right]P_i+\left[\frac{\coth(c)}{c}-\frac{1}{c^2}\right]P_i^\perp\right.\\
&\qquad\qquad\qquad\qquad\left.-
\left[\coth(c)-\frac{1}{c}\right](r_i\sigma_i^T+\sigma_ir_i^T)+
\sigma_i\sigma_i^T\right\},\end{split}\end{equation*}
where $c=\frac{\beta|\sigma^{(i)}|}{n}$.  Recall that if 
$c=o(1)$ then $\coth(c)-\frac{1}{c}\approx\frac{c}{3}$.
In this case,
\begin{equation}\begin{split}\label{quad-diff1}
\E\left[(W_n'-W_n)(W_n'-W_n)^T\big|\sigma\right]&=
\frac{a^2}{n^2}\sum_{i=1}^n\left\{\frac{1}{3}P_i+\frac{1}{3}P_i^\perp-\frac{c}{3}
(r_i\sigma_i^T+\sigma_ir_i^T)+\sigma_i\sigma_i^T\right\}+R''\\
&=\frac{a^2}{3n}Id+\frac{a^2}{n^2}\sum_{i=1}^n
\sigma_i\sigma_i^T-\frac{a^2c}{3n^2}\sum_{i=1}^n(r_i\sigma_i^T+\sigma_ir_i^T)
+R'',.
\end{split}\end{equation}
where the remainder term $R''$ incurred in this approximation is 
\begin{equation*}\begin{split}
R''&=\frac{a^2}{n^2}\sum_{i=1}^n
\left\{\left[\frac{2}{3}-\frac{2c\coth\left(c\right)-2}{c^2}
\right]P_i+\left[\frac{\coth(c)}{c}-\frac{1}{c^2}-\frac{1}{3}\right]
P_i^\perp\right.\\&\qquad\qquad\qquad\qquad\qquad\qquad\qquad\left.-
\left[\coth(c)-\frac{1}{c}-\frac{c}{3}\right](r_i\sigma_i^T+\sigma_ir_i^T)
\right\}.\end{split}\end{equation*}
Observe further that, for the second-last term of \eqref{quad-diff1}, 
using $r_i=\frac{\sigma^{(i)}}{|\sigma^{(i)}|}$ and $c=\frac{\beta|
\sigma^{(i)}|}{n}$ yields 
\begin{equation*}\begin{split}
\frac{a^2c}{3n^2}\sum_{i=1}^n(r_i\sigma_i^T+\sigma_ir_i^T)&=
\frac{a^2\beta}{3n^3}\sum_{i=1}^n\sum_{j\neq i}(\sigma_j\sigma_i^T+\sigma_i
\sigma_j^T)=\frac{2\beta}{3n^2}W_nW_n^T-\frac{2a^2\beta}{3n^3}\sum_{i=1}^n
\sigma_i\sigma_i^T.
\end{split}\end{equation*}

Putting the pieces together,
\begin{equation*}\begin{split}
\E&\left[(W_n'-W_n)(W_n'-W_n)^T\big|\sigma\right]
\\&\qquad=\left(\frac{3-\beta}{3n}\right)\left[Id+\frac{1}{n}\sum_{i=1}^n
3\sigma_i\sigma_i^T\right]-\frac{2\beta}{3n^2} W_nW_n^T+\frac{2\beta}{3n^4}\sum_{i=1}^n\sigma_i\sigma_i^T+R''.
\end{split}\end{equation*}
Note that if $X$ is uniformly
distributed on the sphere, then if $U\in\O_3$, 
\[\E[XX^T]=\E[UXU^TUX^TU^T]=U\E[XX^T]U^T,\]
and thus $\E[XX^T]$ is a scalar matrix.  Moreover,
$\tr(XX^T)=1$ since $XX^T$ is a 
rank-one projection, and thus $\E[\tr(XX^T)]=1$ and so since $\E[XX^T]$
is scalar, it follows that $\E[XX^T]=\frac{1}{3}Id$.  That is,
$\E[\sigma_i\sigma_i^T]=\frac{1}{3}Id$ for each $i$.  We therefore write
\[\E\left[(W_n'-W_n)(W_n'-W_n)^T\big|\sigma\right]=2\left(\frac{1-\frac{\beta}{3}}{n}\right)\Sigma+R',\]
as in Theorem \ref{normal_approx}, with $\Sigma=Id$ and
\[R'=\left(\frac{1-\frac{\beta}{3}}{n}\right)\left[\frac{1}{n}\sum_{i=1}^n
3\sigma_i\sigma_i^T-Id\right]-\frac{2\beta}{3n^2} W_nW_n^T+\frac{2\beta}{3n^4}\sum_{i=1}^n\sigma_i\sigma_i^T+R''.\]
Note in particular that the expected value of the first term of $R'$ is zero.

\end{proof}

\begin{proof}[Proof of Lemma \ref{L:error_bounds}]

Recall that $\Lambda=\left(\frac{1-\frac{\beta}{3}}{n}\right)Id$, and thus 
$\|\Lambda^{-1}\|_{op}=\frac{n}{1-\frac{\beta}{3}}.$  

Now, from Lemma \ref{L:errors},
\[R=-\frac{\beta}{3n^2}W_n+\frac{a}{n^{3/2}}\sum_{i=1}^n\left[\coth\left(\frac{\beta
|\sigma^{(i)}|}{n}\right)-\frac{n}{\beta|\sigma^{(i)}|}-\frac{\beta|\sigma^{
(i)}|}{3n}
\right]\left(\frac{\sigma^{(i)}}{|\sigma^{(i)}|}\right).\]
Note that, while it was previously argued heuristically
that $\E|W_n|^2\approx 3$, one can in fact use the same argument together with the fact
that $\coth(x)-\frac{1}{x}\le\frac{x}{3}$ (shown in the proof of Lemma 
\ref{subcrit}) to show that $\E|W_n|^2\le 3$.  It then follows that 
\[\frac{\beta}{3n^2}\E|W_n|\le\frac{\beta}{3n^2}\sqrt{\E|W_n|^2}\le
\frac{\beta}{\sqrt{3}n^2}.\]
To estimate the second half of $R$, fix $\epsilon=\epsilon(n)\in(0,1)$ 
to be chosen
later.  For notational convenience, let 
$r(t):=\coth(t)-\frac{1}{t}-\frac{t}{3};$ observe that if $t\le\epsilon$
then $|r(t)|<b\epsilon^2$, where $b$ is a universal constant.  Then
the second half of $R$ can be estimated as
\begin{equation}\begin{split}\label{Ept2}
\frac{a}{n^{3/2}}\left|
\sum_{i=1}^n\left[r\left(\frac{\beta|\sigma^{(i)}|}{n}\right)
\right]\left(\frac{\sigma^{(i)}}{|\sigma^{(i)}|}\right)\right|&\le
\frac{a}{n^{3/2}}\left|
\sum_{i=1}^n\left[r\left(\frac{\beta|\sigma^{(i)}|}{n}\right)
\right]\left(\frac{\sigma^{(i)}}{|\sigma^{(i)}|}\right)\1\left(\frac{
\beta|\sigma^{(i)}|}{n}\le\epsilon\right)\right|\\&\qquad\qquad+
\frac{a}{n^{3/2}}\left|
\sum_{i=1}^n\left[r\left(\frac{\beta|\sigma^{(i)}|}{n}\right)
\right]\left(\frac{\sigma^{(i)}}{|\sigma^{(i)}|}\right)\1\left(\frac{
\beta|\sigma^{(i)}|}{n}>\epsilon\right)\right|\\
&\le\frac{ba\epsilon^2}{\sqrt{n}}+\frac{a}{n^{3/2}}\sum_{i=1}^n\1\left(\frac{
\beta|\sigma^{(i)}|}{n}>\epsilon\right),
\end{split}\end{equation}
making use of the fact that $\big|r\left(\frac{\beta|\sigma^{(i)}|}{n}\right)\big|\le1$
for any configuration $\sigma$.
From the LDP for $\sigma^{(i)}$ (i.e., Theorem \ref{spinLDPbeta}), 
\[\P\left[\frac{\beta|\sigma^{(i)}|}{n}>\epsilon\right]\le
C\exp\left[-\frac{n}{2}\inf\{I_\beta(x):x\ge\epsilon\}\right],\]
where 
\[I_\beta(x)=c\coth(c)-1-\log\left(\frac{\sinh(c)}{c}\right)
-\frac{\beta}{2}\left|\coth(c)-\frac{1}{c}\right|^2,\]
and $c$ is the unique element of $\R^+$ such that $|x|=\coth(c)-\frac{1}{c}$.
It is shown in the appendix that $I_\beta(x)$ is increasing for $\beta<3$, and thus
$\inf\{I_\beta(x):x\ge\epsilon\}=I_\beta(\epsilon)$.  Moreover, there
is a universal constant $q>0$ such that for $\epsilon\in(0,1)$,
$I_\beta(\epsilon)\ge\frac{\epsilon^2}{6}\left(1-\frac{\beta}{3}\right)-
q\epsilon^3$.
It follows that
\[\P\left[\frac{\beta|\sigma^{(i)}|}{n}>\epsilon\right]\le
C\exp\left[-\frac{n\epsilon^2}{12}\left(1-\frac{\beta}{3}\right)+
nq\epsilon^3\right].\]
Choose $\epsilon=\epsilon(n)$ such that
$\epsilon^2=\frac{12\log(n)}{n\left(1-\frac{\beta}{3}\right)}$.
Then
\(\P\left[\frac{\beta|\sigma^{(i)}|}{n}>\epsilon\right]\le
\frac{C'}{n},\) and so it follows from  from the bound in \eqref{Ept2} that
\begin{equation*}\begin{split}
\frac{a}{n^{3/2}}\E\left|
\sum_{i=1}^n\left[r\left(\frac{\beta|\sigma^{(i)}|}{n}\right)
\right]\left(\frac{\sigma^{(i)}}{|\sigma^{(i)}|}\right)\right|&\,\le\,
\frac{ba\epsilon^2}{\sqrt{n}}+\frac{a}{\sqrt{n}}\P\left[\frac{\beta|\sigma^{(1)}|}{n}>\epsilon\right]\,\le\,\frac{c_\beta\log(n)}{n^{3/2}}.
\end{split}\end{equation*}
This completes the proof of part (a).

\medskip

For part (b), recall from Lemma \ref{L:errors} that $R'$ is given by
\begin{equation*}\begin{split}
R'&=\left(\frac{1-\frac{\beta}{3}}{n}\right)\left[\frac{1}{n}\sum_{i=1}^n
3\sigma_i\sigma_i^T-Id\right]-\frac{2\beta}{3n^2} W_nW_n^T+\frac{2\beta}{3n^4}\sum_{i=1}^n\sigma_i\sigma_i^T\\&\qquad\qquad+\frac{a^2}{n^2}\sum_{i=1}^n
\left\{\left[\frac{2}{3}-\frac{2c\coth\left(c\right)-2}{c^2}
\right]P_i+\left[\frac{\coth(c)}{c}-\frac{1}{c^2}-\frac{1}{3}\right]
P_i^\perp\right.\\&\qquad\qquad\qquad\qquad\qquad\qquad\qquad\qquad\left.-
\left[\coth(c)-\frac{1}{c}-\frac{c}{3}\right](r_i\sigma_i^T+\sigma_ir_i^T)
\right\}.\end{split}\end{equation*}
For  $x\in\R^n$,
 $\|xx^T\|_{HS}=|x|^2$, and so 
\[\E\|W_nW_n^T\|_{HS}=\E|W_n|^2\le 3\]
and
\[\E\|\sigma_i\sigma_i^T\|_{HS}=\E|\sigma_i|^2=1,\]
which quickly takes care of the middle two terms.

Estimating $\E\|\frac{1}{n}\sum_{i=1}^n(3\sigma_i\sigma_i^T-Id)\|_{HS}$
is a bit more involved.  First, recall that $\|A\|_{HS}=\sqrt{\tr(AA^T)},$
and so by the Cauchy-Schwarz inequality,
\[\E\left\|\frac{1}{n}\sum_{i=1}^n(3\sigma_i\sigma_i^T-Id)\right\|_{HS}\le
\frac{1}{n}\sqrt{\sum_{i,j=1}^n\E\tr\big[(3\sigma_i\sigma_i^T-Id)(3\sigma_j
\sigma_j^T-Id)\big]}.\]
Now,
\begin{equation*}\begin{split}
\E\tr\big[(3\sigma_i\sigma_i^T-Id)^2\big]&=\E\left[9\tr(\sigma_i\sigma_i^T
\sigma_i\sigma_i^T)-6\tr(\sigma_i\sigma_i^T)+Id\right]
=9\E|\sigma_i|^4-6\E|\sigma_i|^2+3=6.
\end{split}\end{equation*}
Similarly, for $i\neq j$,
\begin{equation*}\begin{split}
\E\tr\big[(3\sigma_i\sigma_i^T-Id)(3\sigma_j
\sigma_j^T-Id)\big]&=9\E\left[\inprod{\sigma_i}{\sigma_j}^2\right]-
3.
\end{split}\end{equation*}
Observe that
\begin{equation*}\begin{split}
\E\left[\inprod{\sigma_1}{\sigma_2}^2\big|\{\sigma_i\}_{i\neq 1}\right]&=
\sigma_2^T\E\left[\sigma_1\sigma_1^T\big|\{\sigma_i\}_{i\neq 1}\right]\sigma_2\\
&=\sigma_2^T\left(\frac{1}{Z_1}\int_{\s^2}\theta\theta^T\exp\left[\frac{\beta}{n}
\inprod{\theta}{\sigma^{(1)}}\right]d\mu(\theta)\right)\sigma_2\\
&=\sigma_2^T\left(\left[1-\frac{2c\coth(c)-2}{c^2}\right]P_i+\left[
\frac{c\coth(c)-1}{c^2}\right]P_i^\perp\right)\sigma_2,
\end{split}\end{equation*}
where we have made use of the computation of expression \eqref{theta_part} 
carried out earlier, and again $c=\frac{\beta|\sigma^{(1)}|}{n}$, 
$P_i$ denotes orthogonal projection onto the span of $\sigma^{(1)}$,
and $P_i^\perp$ denotes orthogonal projection onto the orthogonal complement
of the  span of $\sigma^{(1)}$.
Once again making use of the fact that
$c=o(1)$ with high probability, and so $\coth(c)-\frac{1}{c}\approx
\frac{c}{3}-\frac{c^3}{45},$  
\[\E\left[\inprod{\sigma_1}{\sigma_2}^2\big|\{\sigma_i\}_{i\neq 1}\right]
\approx\sigma_2^T\left(\frac{1}{3}Id+\frac{2c^2}{45}P_i-\frac{c^2}{45}
P_i^\perp\right)\sigma_2=\frac{1}{3}-\frac{c^2}{45}+\frac{c^2}{15}
\E\left[\sigma_2^T\sigma_1\sigma_1^T\sigma_2\big|\{\sigma_i\}_{i\neq 1}
\right].\]
Taking expectation of both sides yields
\[\E\left[\inprod{\sigma_1}{\sigma_2}^2\right]\approx
\frac{1}{3}+\E\left[-\frac{ c^2}{45}+\frac{c^2}{15}\inprod{\sigma_1}
{\sigma_2}^2\right],\]
so that
\[\E\tr\big[(3\sigma_i\sigma_i^T-Id)(3\sigma_j\sigma_j^T-Id)\big]\approx
\E\left[-\frac{ c^2}{5}+\frac{3c^2}{5}\inprod{\sigma_1}
{\sigma_2}^2\right].\]

The error incurred in this approximation (for each pair $i\neq j$) is 
\[\E\left|9\sigma_2^T\left(\left[\frac{2}{3}-\frac{2c\coth(c)-2}{c^2}-\frac{2c^2}{45}
\right]P_i+\left[\frac{c\coth(c)-1}{c^2}-\frac{1}{3}+\frac{c^2}{45}\right]P_i^\perp
\right)\sigma_2\right|\le r_1\E c^3,\]
where $r_1$ is a universal constant.

Now, 
\[\E c^2=\frac{\beta^2}{n^2}\sum_{i,j>1}\E\inprod{\sigma_i}{\sigma_j}\le
\frac{\beta^2}{n^2}\left[n-1+(n-1)(n-2)\frac{\beta}{n(3-\beta)}\right]\le
\frac{3\beta^2}{n(3-\beta)},\]
and one can then trivially also estimate that 
\[\E c^3\le\frac{3\beta^3}{n(3-\beta)},\]
and so
\[\E\left\|\frac{1}{n}\sum_{i=1}^n(3\sigma_i\sigma_i^T-Id)\right\|_{HS}\le
\sqrt{\frac{6+\frac{r_2(\beta^2+\beta^3)}{(3-\beta)}}{n}}.\]
The remaining term of the error $R'$ is what was called $R''$ in the
proof of Lemma \ref{L:errors}, for which the remainder from 
Taylor's theorem also suffices:
\begin{equation*}\begin{split}
\E\|R''\|_{HS}&\le\frac{a^2}{n^2}\sum_{i=1}^n\E
\left\{\left\|\left[\frac{2}{3}-\frac{2c\coth\left(c\right)-2}{c^2}
\right]P_i\right\|_{HS}+\left\|\left[\frac{\coth(c)}{c}-\frac{1}{c^2}-\frac{1}{3}\right]
P_i^\perp\right\|_{HS}\right.\\&\qquad\qquad\qquad\qquad\qquad\qquad\qquad\left.-
\left\|\left[\coth(c)-\frac{1}{c}-\frac{c}{3}\right](r_i\sigma_i^T+\sigma_ir_i^T)
\right\|_{HS}\right\}\\
&\le\frac{c_1(3-\beta)\beta}{n^3}\sum_{i=1}^n\E|\sigma^{(i)}|\\
&\le\frac{c_1\sqrt{3(3-\beta)}\beta}{n^{3/2}},\end{split}\end{equation*}
for a universal constant $c_1$, using the facts that $\|P_i\|_{HS}$, $\|P_i^\perp\|_{HS}$ and $\|r_i\sigma_i^T\|_{HS}$
are all bounded by $\sqrt{2}$ or better and that $\E|\sigma^{(i)}|\le\sqrt{\frac{3n}{
3-\beta}}$.  

This completes the proof of part (b).

Finally, part (c) is trivial:
\[\E|W_n'-W_n|^3=\frac{a^3}{n^{3/2}}\E|\sigma_I'-\sigma_I|\le
\frac{8a^3}{n^{3/2}}.\]
\end{proof}

\section{The total spin in the supercritical phase}\label{S:supcrit}

In order to obtain a proof of Theorem \ref{T:supcrit_CLT}, we apply the following
version of Stein's abstract
normal approximation theorem (see \cite{St}, p.\ 35).  The 
formulation below is essentially due to Rinott
and Rotar (\cite{RR}, Thm 1.2), and is a univariate analog of Theorem
\ref{normal_approx} from the the previous section.
\begin{thm}\label{stein-abstract}
Let $h:\R\to\R$ be bounded with bounded derivative.  Suppose that
$(W,W')$ is an exchangeable pair and let $\mathcal{F}$ be a
$\sigma$-field with respect to which $W$ is measureable.  Suppose
further that there is $\lambda>0$ and an
$\mathcal{F}$-measurable random variable $R$ such that
\[\E[W'-W\big|\mathcal{F}]=-\lambda W +R.\]
Then if $Z$ is a centered Gaussian random variable with variance $\sigma^2$, 
\[\big|\E h(W)-\E h(Z)\big|\le \sqrt{\frac{\pi}{2}}\frac{\|h\|_\infty
\E|R|}{\lambda}
+2\|h\|_\infty\E\left|\sigma^2-\frac{1}{2\lambda}
\E\left[(W'-W)^2\big|\mathcal{F}\right]\right|+\frac{\|h'\|_\infty\E|W'-W|^3}{4\lambda}.\]

\end{thm}

\medskip

To apply Theorem \ref{stein-abstract} to $W_n=\sqrt{n}\left[\frac{\beta^2}{n^2k_2^2}\left|\sum_{j=1}^n\sigma_j
\right|^2-1\right]$, we construct an
exchangeable pair $(W_n,W_n')$ using the Gibbs sampler as before; that is,
define $W_n'$ by first replacing a randomly chosen spin in $\{\sigma_i\}$
according to its conditional distribution given the rest of the spins.
The following lemma contains the bounds needed to obtain Theorem
\ref{T:supcrit_CLT} from Theorem \ref{stein-abstract}; with them, the
proof of Theorem \ref{T:supcrit_CLT} is immediate.
\begin{lemma}
For $c_\beta$ a constant depending only on $\beta$, $(W_n,W_n')$ as
constructed above, and $g(x)=\coth(x)-\frac{1}{x}$,
\begin{enumerate}
\item for $\lambda=\frac{(1-\beta  g'(k_2))}{n}$, 
\[\E\left[W_n'-W_n\big|\sigma\right] =-\lambda W_n+R\qquad
and\qquad
\E|R|\le\frac{c_\beta\log(n)}{n^{3/2}};\]
\item for $\sigma^2=\frac{4\beta^2}{\left(1-\beta g'(k_2)\right)k_2^2}
\left[\frac{1}{k_2^2}-\frac{1}{\sinh^2(k_2)}
\right],$
\[\E\left|\sigma^2-\frac{1}{2\lambda}
\E\left[(W_n'-W_n)^2\big|\sigma\right]\right|\le\frac{c_\beta(\log(n))^{1/4}}{n^{1/4}};\]
\item $\E|W_n'-W_n|^3\le\frac{c_\beta}{n^{3/2}}.$
\end{enumerate}
\end{lemma}

\begin{proof}
First note that it is shown
in the Appendix (Lemma \ref{beta_k_2}) that $\beta g'(k_2)=\beta\left(\frac{1}{k_2^2}
-\frac{1}{\sinh^2(k_2)}\right)<1$, so that for $\lambda$ and
$\sigma^2$ as defined above are both strictly positive.

Now,
\begin{equation}\begin{split}\label{E:lindiff1}
\E&\left[W_n'-W_n\big|\sigma\right]\\&=-\frac{\beta^2}{n^{5/2}k_2^2}\sum_{i=1}^n\left[
2\sum_{k\neq i}\inprod{\sigma_i}{\sigma_k}-\E\left[2\sum_{k\neq i}\inprod{\sigma_i
}{\sigma_k}\big|\{\sigma_j\}_{j\neq i}\right]\right]
\\&=-\frac{2\beta^2}{n^{5/2}k_2^2}\left(\left|\sum_{i=1}^n\sigma_i\right|^2-n\right)+\frac{2\beta^2}{n^{5/2}k_2^2}\sum_{i=1}^n\left[\coth\left(\frac{\beta
|\sigma^{(i)}|}{n}\right)-\frac{n}{\beta|\sigma^{(i)}|}
\right]|\sigma^{(i)}|\\&=-\frac{2}{n}W_n-\frac{2}{\sqrt{n}}+\frac{2\beta^2}{n^{3/2}k_2^2}+
\frac{2\beta^2}{n^{5/2}k_2^2}\sum_{i=1}^n\left[\coth\left(\frac{\beta
|\sigma^{(i)}|}{n}\right)-\frac{n}{\beta|\sigma^{(i)}|}
\right]|\sigma^{(i)}|. 
\end{split}\end{equation}

For notational convenience, let $g(x):=\coth(x)-\frac{1}{x}$.  The
first simplification to the expression in \eqref{E:lindiff1} is to observe that
$|\sigma^{(i)}|=|S_n-\sigma_i|$ is close to $|S_n|$ for each $i$; the
error incurred by replacing each
$|\sigma^{(i)}|$ with  $|S_n|$ is estimated as follows.  First,
\[\left|g\left(\frac{\beta|\sigma^{(i)}|}{n}\right)-g\left(\frac{\beta
|S_n|}{n}\right)\right|\le \|g'\|_\infty\left|\frac{\beta|\sigma^{(i)}|}{n}-
\frac{\beta|S_n|}{n}\right|\le\frac{ \|g'\|_\infty\beta}{n},\]
since $S_n-\sigma^{(i)}=\sigma_i,$ which has length 1.  It follows that
\[\frac{2\beta^2}{n^{5/2}k_2^2}\sum_{i=1}^n\E\left[\left|g\left(\frac{\beta
|\sigma^{(i)}|}{n}\right)-g\left(\frac{\beta
|S_n|}{n}\right)\right||\sigma^{(i)}|\right]\le\frac{2\|g'\|_\infty \beta^3}{
n^{3/2}k_2^2},\]
 since $|\sigma^{(i)}|\le n$.
Next, observe that
\[\frac{2\beta^2}{n^{5/2}k_2^2}\sum_{i=1}^n\E\left[\left|g\left(\frac{\beta
|S_n|}{n}\right)\right|\Big||\sigma^{(i)}|-|S_n|\Big|\right]\le\frac{2
\|g\|_\infty\beta^2}{n^{3/2}k_2^2},\]
and so 
\begin{equation*}\begin{split}
\frac{2\beta^2}{n^{5/2}k_2^2}\sum_{i=1}^n&\E\left[\left|g\left(\frac{\beta
|\sigma^{(i)}|}{n}\right)|\sigma^{(i)}|-g\left(\frac{\beta
|S_n|}{n}\right)|S_n|\right|\right]\\&\le\frac{2\beta^2}{n^{5/2}k_2^2}\sum_{i=1}^n\E\left[\left|g\left(\frac{\beta
|\sigma^{(i)}|}{n}\right)-g\left(\frac{\beta
|S_n|}{n}\right)\right||\sigma^{(i)}|+\left|g\left(\frac{\beta
|S_n|}{n}\right)\right|\Big||\sigma^{(i)}|-|S_n|\Big|\right]\\&\le\frac{2\|g'\|_\infty\beta^3+2\|g\|_\infty\beta^2}{n^{3/2}k_2^2};\end{split}\end{equation*}
that is,
\begin{equation}\begin{split}\label{E:first-improvement}
\E\left[W_n'-W_n\big|\sigma\right]&=-\frac{2}{n}W_n-\frac{2}{\sqrt{n}}+
\frac{2\beta^2}{n^{3/2}k_2^2}g\left(\frac{\beta|S_n|}{n}\right)|S_n|+R_1, 
\end{split}\end{equation}
where $\E|R_1|\le\frac{c_\beta}{n^{3/2}}.$

We next approximate $g\left(\frac{\beta|S_n|}{n}\right)$
by a first-order Taylor polynomial, making use of the LDP for
$|S_n|$ (Theorem \ref{spinLDPbeta}).  We have that

\begin{equation}\label{norm_LDP}
\limsup_{n\to\infty}\frac{1}{n}P_{n,\beta}\left[\left|
\frac{|S_n|}{n}-\frac{k_2}{\beta}\right|\ge\epsilon\right]\le
-\inf_{\left||x|-\frac{k_2}{\beta}\right|\ge\epsilon} I_\beta(x),\end{equation}
where
\[I_\beta(x)=\Phi_\beta(y)-\varphi(\beta),\qquad
\Phi_\beta(y)=y|x|+\log\left(\frac{y}{\sinh(y)}
\right)-\frac{\beta}{2}|x|^2,\]
and $y$ is uniquely defined by
\[\coth(y)-\frac{1}{y}=|x|.\]
Recall also that 
\[\varphi(\beta)=\inf_{x\ge 0}\Phi_\beta(y)=\frac{k_2^2}{2\beta}+\log\left(
\frac{k_2}{\sinh(k_2)}\right).\]
It was moreover shown in the Appendix (Lemmas \ref{supcrit} and
\ref{beta_k_2}) that $|x|=\frac{k_2}{\beta}$ corresponding
to $y=k_2$ is the unique minimizing set for $\Phi_\beta$, and that 
$\Phi_\beta(y)$ is decreasing as a function of $|x|$ on $\left[0,\frac{k_2}{
\beta}\right]$ and increasing on $\left[\frac{k_2}{\beta},\infty\right)$.
This means in particular that 
\[\inf_{\left||x|-\frac{k_2}{\beta}\right|\ge\epsilon} I_\beta(x)=\min \left\{I_\beta\left(
y\left(\frac{k_2}{\beta}+\epsilon\right)\right),I_\beta\left(
y\left(\frac{k_2}{\beta}-\epsilon\right)\right)\right\},\]
where by $I_\beta(y(t))$ we mean the value that $I_\beta(y)$ takes on for 
all $x$ with $|x|=t$ and $y$ defined in terms of $|x|$ as above.

Now, we know that $\Phi_\beta(y)$ is minimized on $|x|=\frac{k_2}{\beta}$,
so that $\Phi_\beta'\left(k_2\right)=0$. It is shown in the Appendix
(Lemma \ref{beta_k_2})
 that $\Phi_\beta''\left(k_2\right)> 0$, so
 that there is a constant $K_{\beta}$ such that
\[\inf_{\left||x|-\frac{k_2}{\beta}\right|\ge\epsilon} I_\beta(x)\ge K_\beta \epsilon^2,\]
and so
\[P_{n,\beta}\left[\left|\frac{|S_n|}{n}-\frac{k_2}{\beta}
\right|\ge\epsilon\right]\le e^{-K_\beta n\epsilon^2}.\]

Applying this estimate then yields
\begin{equation*}\begin{split}
\E&\left|\frac{2\beta^2}{n^{3/2}k_2^2}\left[g\left(\frac{\beta
    |S_n|}{n}\right)-g\left(k_2\right)-g'(k_2)\left(\frac{\beta
    |S_n|}{n}-k_2\right)
  \right]|S_n|\right|\\&\qquad\qquad\le\frac{2\beta^2}{n^{3/2}k_2^2}
\|g''\|_\infty\epsilon^2\E\left||S_n|\1\left(\left|
\frac{\beta|S_n|}{n}-k_2
\right|\le\epsilon\right)\right|\\&\qquad\qquad\qquad+
\frac{2\beta^2}{n^{3/2}k_2^2}
\big[2\|g\|_\infty+\|g'\|_\infty(\beta+k_2)\big]\E\left||S_n|\1\left(\left|
\frac{\beta|S_n|}{n}-k_2
\right|\ge\epsilon\right)\right|\\&\qquad\qquad\le
\frac{2\beta^2}{\sqrt{n}k_2^2}\left[C\epsilon^2+C'
e^{-K_\beta n\epsilon^2}\right],
\end{split}\end{equation*}
where we have also used the trivial estimate $|S_n|\le n$ in the last
line.
Choosing $\epsilon^2=\frac{\log(n)}{K_\beta n}$ gives that 
\[\E\left|\frac{2\beta^2}{n^{3/2}k_2^2}\left[g\left(\frac{\beta
    |S_n|}{n}\right)-g\left(k_2\right)-g'(k_2)\left(\frac{\beta
    |S_n|}{n}-k_2\right)
  \right]|S_n|\right|\le c_\beta\frac{\log(n)}{n^{3/2}},\]
for some constant $c_\beta$ depending only on $\beta$.

Combining this estimate with \eqref{E:first-improvement} and recalling
that $g(k_2)=\frac{k_2}{\beta}$ gives that
\begin{equation}\begin{split}\label{E:second-improvement}
\E\left[W_n'-W_n\big|\sigma\right]
&=-\frac{2}{n}W_n+\frac{2}{\sqrt{n}}\left[\frac{\beta|S_n|}{nk_2}-1\right]
+\frac{2\beta^2g'(k_2)}{n^{3/2}k_2^2}\left(\frac{
\beta|S_n|}{n}-k_2\right)|S_n|+R_2, 
\end{split}\end{equation}
where $\E|R_2|\le\frac{c_\beta\log(n)}{n^{3/2}}.$

Now, observe that 
\[|S_n|=\frac{nk_2}{\beta}\sqrt{1+\frac{W_n}{\sqrt{n}}}=\frac{nk_2}{\beta}
\left(1+\frac{W_n}{2\sqrt{n}}+R_3\right),\]
where $|R_3|\le \frac{C|W_n|^2}{n}.$  This means that the second term
of \eqref{E:second-improvement} is
\[\frac{W_n}{n}+\frac{2R_3}{\sqrt{n}},\]
and $\E\left|\frac{2R_3}{\sqrt{n}}\right|\le
\frac{C}{n^{3/2}}\E|W_n|^2\le \frac{C\log(n)}{n^{3/2}},$
using the LDP as above.  
Similarly, 
\[\frac{2\beta^2g'(k_2)}{n^{3/2}k_2^2}\E\left|\left(\frac{
\beta|S_n|}{n}-k_2\right)|S_n|-\frac{k_2W_n}{2\sqrt{n}}|S_n|\right|
\le \frac{c_\beta }{n^{5/2}}\E\left[|W_n|^2|S_n|\right]]\le 
\frac{c_\beta \log(n)}{n^{3/2}},\]
and 
\[\frac{2\beta^2g'(k_2)}{n^{3/2}k_2^2}\E\left|\frac{k_2W_n}{2\sqrt{n}}|S_n|
-\frac{\sqrt{n}k_2^2W_n}{2\beta}\right|\le\frac{2\beta^2g'(k_2)}{n^{3/2}k_2^2}
\E\left[\frac{Ck_2^2|W_n|^2}{4\beta}\right]\le
\frac{c_\beta\log(n)}{n^{3/2}}, \]
again using that $\E|W_n|^2\le\log(n)$.
Combining these estimates with \eqref{E:second-improvement} yields
\begin{equation}\begin{split}\label{E:third-improvement}
\E\left[W_n'-W_n\big|\sigma\right] &=-\frac{(1-\beta
  g'(k_2))}{n}W_n+R_4,\end{split}\end{equation} where again
$\E|R_4|\le\frac{c_\beta\log(n)}{n^{3/2}}.$ 
This completes the proof of part $(a)$.

\medskip

For part $(b)$, observe that by definition,
\begin{equation}\begin{split}\label{Delta^2-from-def}
\E\left[(W_n'-W_n)^2\big|\sigma\right]&=\frac{\beta^4}{n^4k_2^4}\sum_{i=1}^n\E
\left[\left.\left(2\sum_{j\neq i}\inprod{\sigma_i^*-\sigma_i}{\sigma_j}
\right)^2\right|\sigma,I=i\right]\\&=\frac{4\beta^4}{n^4k_2^4}\sum_{i=1}^n
\sum_{j,k\neq i}\E\left[\sigma_j^T(\sigma_i^*-\sigma_i)
(\sigma_i^*-\sigma_i)^T\sigma_k\big|\sigma,I=i\right].
\end{split}\end{equation}
Now, $\E\left[(\sigma_i^*-\sigma_i)
(\sigma_i^*-\sigma_i)^T\big|\sigma,I=i\right]$ was already
computed exactly in the $\beta<3$ case, and was found to be
\begin{equation*}\begin{split}\E\left[(\sigma_i^*-\sigma_i)
(\sigma_i^*-\sigma_i)^T\big|\sigma,I=i\right]&=
\left\{\left[1-\frac{2c_i\coth\left(c_i\right)-2}{c_i^2}
\right]P_i+\left[\frac{\coth(c_i)}{c_i}-\frac{1}{c_i^2}\right]P_i^\perp\right.\\
&\qquad\qquad\qquad\qquad\left.-
\left[\coth(c_i)-\frac{1}{c_i}\right](r_i\sigma_i^T+\sigma_ir_i^T)+
\sigma_i\sigma_i^T\right\},\end{split}\end{equation*}
where $c_i=\frac{\beta|\sigma^{(i)}|}{n}$, $r_i=\frac{
\sigma^{(i)}}{|\sigma^{(i)}|}$, and $P_i$ is orthogonal projection onto
$r_i$.

Recall that 
\[P_{n,\beta}\left[\left|\frac{\beta|\sigma^{(i)}|}{n}-\frac{k_2(n-1)}{n}
\right|\ge\epsilon\right]\le e^{-K_\beta n\epsilon^2}\]
and that
\[\coth(k_2)-\frac{1}{k_2}=\frac{k_2}{\beta}.\]
Using this above,
\begin{equation}\label{E-Delta^2-approx}
\begin{split}\E\left[(\sigma_i^*-\sigma_i)
(\sigma_i^*-\sigma_i)^T\big|\sigma,I=i\right]&=
\left(1-\frac{2}{\beta}\right)P_i+\frac{1}{\beta}P_i^\perp-
\frac{k_2}{\beta}(r_i\sigma_i^T+\sigma_ir_i^T)+
\sigma_i\sigma_i^T+R_i'\\
&=\frac{1}{\beta}Id+\left(1-\frac{3}{\beta}\right)P_i-
\frac{k_2}{\beta}(r_i\sigma_i^T+\sigma_ir_i^T)+
\sigma_i\sigma_i^T+R_i',
\end{split}\end{equation}
where
\[R_i'=\left(\frac{g(c_i)}{c_i}-\frac{1}{\beta}\right)Id+\left(\frac{3g
(c_i)}{c_i}-\frac{3}{\beta}\right)P_i-\left(g(c_i)-\frac{k_2}{\beta}\right)
\left(r_i\sigma_i^T+\sigma_ir_i^T\right).\]
Ignoring the $R_i'$ for the moment and putting the main term of 
\eqref{E-Delta^2-approx} into \eqref{Delta^2-from-def} yields
\begin{equation*}\begin{split}
\frac{4\beta^4}{n^4k_2^4}\sum_i\sum_{j,k\neq i}&\left[
\frac{1}{\beta}\inprod{\sigma_j}{\sigma_k}+\left(1-\frac{3}{\beta}\right)
\sigma_j^TP_i\sigma_k-
\frac{k_2}{\beta}(\sigma_j^Tr_i\sigma_i^T\sigma_k+\sigma_j^T\sigma_ir_i^T
\sigma_k)+\sigma_j^T\sigma_i\sigma_i^T\sigma_k\right].
\end{split}\end{equation*}
The first term is
\[\frac{4\beta^3}{n^4k_2^4}\sum_i\sum_{j,k\neq i}\inprod{\sigma_j}{\sigma_k}=
\frac{4\beta^3}{n^4k_2^4}\sum_{i=1}^n|\sigma^{(i)}|^2.\]
For the second term, note that
\(\sigma_j^TP_i\sigma_k=\tr(\sigma_k\sigma_j^Tr_ir_i^T),\)
and so
\[\sum_{j,k\neq i}\sigma_j^TP_i\sigma_k
=\tr\left(\sigma^{(i)}[\sigma^{(i)}]^T
r_ir_i^T\right)=\inprod{\sigma^{(i)}}{r_i}^2=|\sigma^{(i)}|^2\]
and thus
\[\frac{4\beta^4}{n^4k_2^4}\sum_i\sum_{j,k\neq i}\left(1-\frac{3}{\beta}\right)
\sigma_j^TP_i\sigma_k=\frac{4\beta^4}{n^4k_2^4}\left(1-\frac{3}{\beta}\right)
\sum_i|\sigma^{(i)}|^2.\]
Similarly, for the first half of the third term,
\[-\frac{4\beta^3}{n^4k_2^3}\sum_i\sum_{j,k\neq i}\sigma_j^Tr_i\sigma_i^T\sigma_k
=-\frac{4\beta^3}{n^4k_2^3}\sum_i|\sigma^{(i)}|\inprod{\sigma^{(i)}}{\sigma_i},\]
and the second half is the same.
Finally, 
\[\frac{4\beta^4}{n^4k_2^4}\sum_i\sum_{j,k\neq i}\sigma_j^T\sigma_i\sigma_i^T
\sigma_k=\frac{4\beta^4}{n^4k_2^4}\sum_i\inprod{\sigma_i}{\sigma^{(i)}}^2;\]
all together,

\begin{equation*}\begin{split}
\E\left[(W_n'-W_n)^2\big|\sigma\right]&=\frac{4\beta^4}{n^4k_2^4}\sum_i
\left[\left(1-\frac{2}{\beta}\right)|\sigma^{(i)}|^2-\frac{2k_2}{\beta}
|\sigma^{(i)}|\inprod{\sigma_i}{\sigma^{(i)}}+\inprod{\sigma_i}{
\sigma^{(i)}}^2\right]\\&\qquad\qquad
+\frac{4\beta^4}{n^4k_2^4}\sum_i\sum_{j,k\neq i}\sigma_j^TR_i'\sigma_k.
\end{split}\end{equation*}

In order to apply Theorem \ref{stein-abstract}, one must recognize in this expression
a deterministic part (which is then called $2\lambda\sigma^2$, from
which $\sigma^2$ is then determined), plus a mean zero part.  
With this motivation in mind, we write
\begin{equation}\begin{split}\label{2nd-diff-master}
\E&\left[(W_n'-W_n)^2\big|\sigma\right]\\&=\frac{4\beta^4}{n^3k_2^4}
\left[2\left(1-\frac{2}{\beta}\right)\frac{(n-1)^2k_2^2
}{\beta^2}-\frac{2n^2k_2^4}{\beta^4}
\right]+\frac{4\beta^4}{n^4k_2^4}\sum_i
\left(1-\frac{2}{\beta}\right)\left(|\sigma^{(i)}|^2-\frac{(n-1)^2k_2^2
}{\beta^2}\right)\\&\qquad\qquad+\frac{4\beta^4}{n^4k_2^4}\sum_i\left[
-\frac{2k_2}{\beta}\left(|\sigma^{(i)}|\inprod{\sigma_i}{
\sigma^{(i)}}-\frac{n^2k_2^3}{\beta^3}\right)+\inprod{\sigma_i}{
\sigma^{(i)}}^2-\left(1-\frac{2}{\beta}\right)\frac{(n-1)^2k_2^2
}{\beta^2}\right]\\&\qquad\qquad\qquad
+\frac{4\beta^4}{n^4k_2^4}\sum_i\sum_{j,k\neq i}\sigma_j^TR_i'\sigma_k,
\end{split}\end{equation}
and define $\sigma$ such that, to top order in $n$,
\[\frac{4\beta^4}{n^3k_2^4}
\left[2\left(1-\frac{2}{\beta}\right)\frac{(n-1)^2k_2^2
}{\beta^2}-\frac{2n^2k_2^4}{\beta^4}
\right]=2\lambda\sigma^2,\]
for $\lambda=\frac{1-\beta g'(k_2)}{n}$ as above.
Note that the top order in $n$ of the 
expression inside the parentheses can be simplified to
\begin{equation*}\begin{split}
\frac{2n^2k_2^2}{\beta^2}\left[1-\frac{2}{\beta}-\frac{k_2^2}{\beta^2}\right]
&=\frac{2n^2k_2^2}{\beta^2}\left[1-\frac{2\left(\coth(k_2)-\frac{1}{k_2}
\right)}{k_2}-\left(\coth(k_2)-\frac{1}{k_2}\right)^2\right]\\&=
\frac{2n^2k_2^2}{\beta^2}\left[1-\frac{\cosh^2(k_2)}{\sinh^2(k_2)}+\frac{1}{k_2^2}\right]\\&=\frac{2n^2k_2^2}{\beta^2}\left[\frac{1}{k_2^2}-\frac{1}{\sinh^2(k_2)}
\right]>0,
\end{split}\end{equation*}
so defining $\sigma$ in this way does in fact yield a strictly positive
value of $\sigma^2$ which depends only on $\beta$ and is independent of $n$.

Then to apply Theorem \ref{stein-abstract} it is 
necessary to estimate the expected absolute value of each of the 
terms above (except the $2\lambda\sigma^2$ part).  

For the first term, it follows as before from the LDP for $|\sigma^{(i)}|$ that
\[\E\left||\sigma^{(i)}|^2-\frac{(n-1)^2k_2^2
}{\beta^2}\right|\le (n-1)^2\left[\epsilon+e^{-K_\beta(n-1)\epsilon^2}\right];\]
taking $\epsilon=\sqrt{\frac{\log(n)}{K_\beta(n-1)}}$ shows that
\[\E\left||\sigma^{(i)}|^2-\frac{(n-1)^2k_2^2
}{\beta^2}\right|\le c_\beta\sqrt{\log(n)}n^{3/2}.\]
Next, observe that
\[\left|\inprod{\sigma_i}{\sigma^{(i)}}|\sigma^{(i)}|-
\inprod{\sigma_i}{S_n}|S_n|\right|\le\left|\inprod{\sigma_i}{\sigma^{(i)}}
\right|+|S_n|\le 2n.\]
Moreover,
\[\frac{4\beta^4}{n^4k_2^4}\E\left|\sum_i \frac{2k_2}{\beta}
\left(|S_n|\inprod{\sigma_i}{S_n}-\frac{n^2k_2^3}{\beta^3}\right)
\right|=\frac{8\beta^3}{n^4k^3}\E\left||S_n|^3-\frac{n^3k_2^3}{\beta^3}\right|
\le\frac{c_\beta\sqrt{\log(n)}}{n^{3/2}}.\]

Now, using the same definitions for $g,c_i,r_i$ and $P_i$ as before,
\begin{equation*}\begin{split}
\E\left[\inprod{\sigma_i}{\sigma^{(i)}}^2\right]&=
\E\left[\sum_{j,k\neq i}\sigma_j^T\E\left[\sigma_i\sigma_i^T\big|
\{\sigma_{\ell}\}_{\ell\neq i}\right]\sigma_k\right]\\&=
\E\left[\sum_{j,k\neq i}\sigma_j^T\left[\left(1-\frac{3g(c_i)}{c_i}
\right)P_i+\frac{g(c_i)}{c_i}Id\right]\sigma_k\right]=
\E\left[\left(1-\frac{2g(c_i)}{c_i}\right)|\sigma^{(i)}|^2\right];
\end{split}\end{equation*}
this explains the choice of constants in \eqref{2nd-diff-master}.
Furthermore,
\begin{equation}\begin{split}\label{nasty-term}
\E&\left[\sum_i\left(\inprod{\sigma_i}{\sigma^{(i)}}^2-\left(1-\frac{2}{\beta}
\right)\frac{(n-1)^2k_2^2}{\beta^2}\right)\right]^2\\&=n\E\left(
\inprod{\sigma_1}{\sigma^{(1)}}^2-\left(1-\frac{2}{\beta}
\right)\frac{(n-1)^2k_2^2}{\beta^2}\right)^2\\&\qquad+n(n-1)\E\left(
\inprod{\sigma_1}{\sigma^{(1)}}^2-\left(1-\frac{2}{\beta}
\right)\frac{(n-1)^2k_2^2}{\beta^2}\right)\left(\inprod{\sigma_2}{
\sigma^{(2)}}^2-\left(1-\frac{2}{\beta}\right)\frac{(n-1)^2k_2^2}{\beta^2}\right).
\end{split}\end{equation}
The first term will simply be estimated by $c_\beta n^5$.  For the second,
first let $\sigma^{(1,2)}:=\sum_{j>2}\sigma_j$ and observe that
\[\left|\E\inprod{\sigma_1}{\sigma^{(1)}}^2\inprod{\sigma_2}{\sigma^{(2)}}^2
-\E\inprod{\sigma_1}{\sigma^{(1)}}^2\inprod{\sigma_2}{\sigma^{(1,2)}}^2\right|
\le Cn^3.\]  Now,
\begin{equation*}\begin{split}
\E\inprod{\sigma_1}{\sigma^{(1)}}^2\inprod{\sigma_2}{\sigma^{(1,2)}}^2&=
\E\left[\sum_{i,j>1}\sum_{k,\ell>2}\sigma_i^T\sigma_1\sigma_1^T\sigma_j
  \sigma_k^T\sigma_2\sigma_2^T\sigma_\ell\right]\\&=
\E\left[\sum_{i,j>1}\sum_{k,\ell>2}\sigma_i^T\E\left[
\sigma_1\sigma_1^T\big|\{\sigma_m\}_{
      m>1}\right] \sigma_j
  \sigma_k^T\sigma_2\sigma_2^T\sigma_\ell\right]\\&=
\E\left[\left(1-\frac{2g(c_1)}{c_1}\right)|\sigma^{(1)}|^2
\inprod{\sigma_2}{\sigma^{(1,2)}}^2\right]
\end{split}\end{equation*}
By the LDP for $\sigma^{(1)}$, we can replace $\frac{g(c_1)}{c_1}$
by $\frac{1}{\beta}$ and $|\sigma^{(1)}|$ by $\frac{(n-1)k_2}{\beta}$, 
incurring an error of size $c_\beta \sqrt{\log(n)}n^{7/2}.$
At this point we are left with
\[\E\left[\left(1-\frac{2}{\beta}\right)\frac{(n-1)^2k_2^2}{\beta^2}
\inprod{\sigma_2}{\sigma^{(1,2)}}^2\right].\]
We can now go back to $\sigma^{(2)}$ instead of $\sigma^{(1,2)}$
(with another loss of order $n^3$), and therefore replace the last
expression with
\[\E\left[\left(1-\frac{2}{\beta}
\right)\frac{(n-1)^2k_2^2}{\beta^2}\left(1-\frac{2g(c_2)}{c_2}\right)
|\sigma^{(2)}|^2\right],\]
(using the expression for $\E[\inprod{\sigma_i}{\sigma^{(i)}}]$
obtained above).
One final application of the LDP for $\sigma^{(2)}$ now means that, with
loss of the same order as before, this expression is equal to
\[\left(1-\frac{2}{\beta}
\right)^2\frac{(n-1)^4k_2^4}{\beta^4}.\]
Using these approximations in \eqref{nasty-term} now yields
\[\E\left[\sum_i\left(\inprod{\sigma_i}{\sigma^{(i)}}^2-\left(1-\frac{2}{\beta}
\right)\frac{(n-1)^2k_2^2}{\beta^2}\right)\right]^2\le
c_\beta\sqrt{\log(n)}n^{11/2},\]
and so 
\[\frac{4\beta^4}{n^4k_2^4}\E\left|\sum_i\left(\inprod{\sigma_i}{\sigma^{(i)}}^2-\left(1-\frac{2}{\beta}
\right)\frac{(n-1)^2k_2^2}{\beta^2}\right)\right|\le
c_\beta(\log(n))^{1/4}n^{-5/4}.\]

Exactly the same sorts of arguments using  LDP for $\sigma^{(i)}$ also imply that all the errors from the 
$R_i'$ terms are smaller than those already accounted for.

Finally, part $(c)$ is trivial: 
\[\E|W_n'-W_n|^3=\frac{8\beta^6}{n^{9/2}k_2^6}\E\left|\sum_{j\neq I}\inprod{\sigma_I^*
-\sigma_I}{\sigma_j}\right|^3\le\frac{8\beta^6}{n^{3/2}k_2^6}.\]

\end{proof}

\section{The critical temperature}\label{S:critical}

As in the previous section, we begin with an discussion of the correct
normalization for $S_n$ in the current regime.
Recall that the LDP for $|S_n|$ (Theorem \ref{spinLDPbeta}) implies in particular that
\[\P\left[\frac{3|S_n|}{n}>\epsilon\right]\le
C\exp\left[-\frac{n}{2}\inf\{I_3(x):x\ge\epsilon\}\right],\]
where 
\[I_3(x)=c\coth(c)-1-\log\left(\frac{\sinh(c)}{c}\right)
-\frac{3}{2}\left|\coth(c)-\frac{1}{c}\right|^2,\]
and $c$ is the unique element of $\R^+$ such that $|x|=\coth(c)-\frac{1}{c}$.
It is shown in the appendix that $I_3(x)$ is increasing, and thus
$\inf\{I_3(x):x\ge\epsilon\}=I_3(\epsilon)$.  Moreover, there
is a universal constant $r>0$ such that for $\epsilon\in(0,1)$,
$I_3(\epsilon)\ge r\epsilon^4$ (note in particular the difference from the
subcritical case).
It follows that
\[\P\left[\frac{3|S_n|}{n}>\epsilon\right]\le
C\exp\left[-cn\epsilon^4\right].\]
for some constant $c>0$, and so
\[\E\left[\frac{3|S_n|}{n}\right]\le\epsilon+3
C\exp\left[-cn\epsilon^4\right].\]
Choosing $\epsilon^4=\frac{\log(n)}{cn}$ shows that $\E[|S_n|]\le
\frac{C\sqrt[4]{\log(n)}}{n^{3/4}}.$
This at least suggests what turns out to be the correct normalization
for the total spin.  
In fact, one can use a modification of the argument given in
the beginning of section \ref{S:subcrit} to show that 
\[\E|S_n|^2=n+\E\left[|S_n|^2-\frac{|S_n|^4}{5n^2}-\frac{2|S_n|^2}{n}+\frac{2\cdot
  3^5|S_n|^6}{945n^4}+lower\ order\ terms\right],\]
from which it eventually follows that, to top order in $n$, there is a $c_3$ such that
\[\E|S_n|^2=\frac{n^{3/2}}{c_3}\qquad\qquad\E|S_n|^4=5n^3.\]

\medskip

The proof of the limit theorem for $W_n=\frac{C_3|S_n|^2}{n^{3/2}}$  is essentially
via the so-called
``density approach'' to Stein's method introduced by Stein, Diaconis,
Holmes and Reinert \cite{SDHR}; see also the recent work of 
Chatterjee and Shao \cite{CS}, with an application to the
total spin of the mean-field
Ising model (i.e., the Curie-Weiss model).  The following theorem provides the
framework we use for the approximation; the proof is given in Section
\ref{Stein_app} of the Appendix.

\begin{thm}\label{T:abstract_approx}
Let $(W,W')$ be an exchangeable pair of positive random variables.
Suppose there exists a $\sigma$-field $\mathcal{F}\supseteq\sigma(W)$,
$\mathcal{F}$-measurable random variables $R$ and $R'$ and $k>0$  deterministic such that
\[\E\big[W'-W\big|\mathcal{F}\big]=3k\big(1-cW^2\big)+R\]
and
\[\E\big[(W'-W)^2\big|\mathcal{F}\big]=kW+R'.\]
Let $X$ have density 
\[p(t)=\begin{cases}\frac{1}{z}x^5e^{-\frac{ct^2}{2}}&t\ge
  0;\\0&t<0.\end{cases}\]
Then there are constants $C_1,C_2,C_3$  depending only on $c$ such that for all $h\in C^2(\R)$,
\begin{equation*}\begin{split}
\big|\E h(W)-\E h(X)\big|&\le
\frac{C_1\|h\|_\infty}{k}\E|R|+\left(\frac{C_2(\|h\|_\infty+
\|h'\|_\infty)}{k}\right)\E|R'|\\&\qquad+
\left(\frac{C_3(\|h\|_\infty+\|h'\|_\infty+
\|h''\|_\infty)}{k}\right)\E|W'-W|^3.\end{split}\end{equation*}
\end{thm}

Within this framework, we proceed as before: recall that we have defined
$W_n=\frac{c_3}{n^{3/2}}\sum_{i,j=1}^n\inprod{\sigma_i}{\sigma_j},$ and
make an exchangeable pair $(W_n,W_n')$ by replacing a random spin
using the Gibbs sampler. 
The following lemma gives the bounds needed to apply Theorem
\ref{T:abstract_approx} in this setting.
\begin{lemma}\label{L:crit_errors}
There is a universal constant $C$ such that for $(W_n,W_n')$ as constructed above, $k=\frac{2c_3}{3n^{3/2}}$ and
$c=\frac{1}{5c_3^2}$, 
\begin{enumerate}
\item $\E\left[W_n'-W_n\big|\sigma\right]=3k\left(1-cW_n^2\right)+R$
  and \, $\E|R|\le\frac{C\log(n)}{n^2}$;
\item $\E\left[(W_n'-W_n)^2\big|\sigma\right]=kW_n+R',$ and \,
  $\E|R'|\le\frac{C\log(n)}{n^2}$;
\item $\E|W_n'-W_n|^3\le\frac{C\log(n)}{n^{9/4}}.$
\end{enumerate}
\end{lemma}

The proof of Theorem \ref{T:limit_crit} is now immediate.  To prove Theorem
\ref{T:bldist}, the same smoothing argument as in the subcritical case can be
carried out; again, see Section 3 of \cite{MM} for a detailed example.

\begin{proof}[Proof of Lemma \ref{L:crit_errors}]

For part $(a)$,
\begin{equation}\begin{split}\label{diff1_crit}
\E\left[W_n'-W_n\big|\sigma\right]&=-\frac{c_3}{n^{5/2}}\sum_{i=1}^n\left[
2\sum_{k\neq i}\inprod{\sigma_i}{\sigma_k}-\E\left[2\sum_{k\neq i}\inprod{\sigma_i
}{\sigma_k}\big|\{\sigma_j\}_{j\neq i}\right]\right]
\\&=-\frac{2}{n}W_n+\frac{2c_3}{n^{3/2}}+\frac{2c_3}{n^{5/2}}\sum_{i=1}^ng\left(\frac{3
|\sigma^{(i)}|}{n}\right)|\sigma^{(i)}|
\end{split}\end{equation}
just as in the supercritical case,
again using the notation $g(x)=\coth(x)-\frac{1}{x}$.
Now, near zero, $g(x)=\frac{x}{3}-\frac{x^3}{45}+O(x^5).$ Using this
above,
\[\sum_{i=1}^ng\left(\frac{3
|\sigma^{(i)}|}{n}\right)|\sigma^{(i)}|=\sum_{i=1}^n\left[\frac{|\sigma^{(i)}|^2}{n}-
\frac{|\sigma^{(i)}|^4}{5n^3}+O\left(\frac{|\sigma^{(i)}|^6}{n^5}\right)\right].\]
Note that
\[\frac{1}{n}\sum_{i=1}^n
|\sigma^{(i)}|^2=\sum_{j,k}\inprod{\sigma_j}{\sigma_k}-\frac{2}{n}\sum_{i,j}
\inprod{\sigma_i}{\sigma_j}+1=\frac{n^{3/2}W_n}{c_3}-\frac{2\sqrt{n}W_n}{c_3}+1.\]
Similarly, 
\[-\frac{1}{5n^3}\sum_{i=1}^n|\sigma^{(i)}|^4=-\frac{nW_n^2}{5c_3^2}+\frac{4W_n^2}{5c_3^2}-\frac{4\sum_i\inprod{
\sigma_i}{S_n}^2}{5n^3}-\frac{2W_n}{5c_3\sqrt{n}}+\frac{W_n}{5n^{3/2}c_3}+\frac{1}{5n^2}.\]
Using these expressions in \eqref{diff1_crit} yields
\[\E\left[W_n'-W_n\big|\sigma\right]=\frac{2c_3}{n^{3/2}}-\frac{2W_n^2}{5c_3n^{3/2}}+R,\]
where 
\[R=-\frac{4W_n}{n^2}+\frac{8W_n^2}{5c_3n^{5/2}}+\frac{CW_n^3}{n^2c_3^2}+\frac{\tilde{R}(\sigma)}{n^{5/2}},\]
$C$  is a universal constant, and furthermore, $\tilde{R}(\sigma)\le C$ almost
surely.  Note in particular the cancellation of the $\frac{2}{n}W_n$
terms, which is the crucial difference from the subcritical case, and
the reason that the limiting distribution of the total spin is not
Gaussian for $\beta=3$.

Recall that the LDP for $\mu_{\sigma,n}$ gives us that $\E W_n, \E
W_n^2,\E W_n^3\le\log(n),$ and so it follows that $\E|R|\le
\frac{C\log(n)}{n^2}.$
This completes the proof of $(a)$.

\medskip

For part $(b)$, from the
definition as before,
\begin{equation}\begin{split}\label{Delta^2-from-def-crit}
\E\left[(W_n'-W_n)^2\big|\sigma\right]&=\frac{c_3^2}{n^4}\sum_{i=1}^n
\sum_{j,k\neq i}\E\left[\sigma_j^T(\sigma_i^*-\sigma_i)
(\sigma_i^*-\sigma_i)^T\sigma_k\big|\sigma,I=i\right].
\end{split}\end{equation}
Now, $\E\left[(\sigma_i^*-\sigma_i)
(\sigma_i^*-\sigma_i)^T\big|\sigma,I=i\right]$ was 
computed exactly in the $\beta<3$ case, and was found to be
\begin{equation*}\begin{split}\E\left[(\sigma_i^*-\sigma_i)
(\sigma_i^*-\sigma_i)^T\big|\sigma,I=i\right]&=
\left\{\left[1-\frac{2c_i\coth\left(c_i\right)-2}{c_i^2}
\right]P_i+\left[\frac{\coth(c_i)}{c_i}-\frac{1}{c_i^2}\right]P_i^\perp\right.\\
&\qquad\qquad\qquad\qquad\left.-
\left[\coth(c_i)-\frac{1}{c_i}\right](r_i\sigma_i^T+\sigma_ir_i^T)+
\sigma_i\sigma_i^T\right\},\end{split}\end{equation*}
where $c_i=\frac{\beta|\sigma^{(i)}|}{n}$, $r_i=\frac{
\sigma^{(i)}}{|\sigma^{(i)}|}$, $P_i$ is orthogonal projection onto
$r_i$, and $P_i^\perp$ is orthogonal projection onto the orthogonal
complement of $r_i$.

Recall that $g(x)=\coth(x)-\frac{1}{x}\approx\frac{x}{3}$ for
small $x$.  Using this above,
\begin{equation}\label{E-Delta^2-approx-crit}
\begin{split}\E\left[(\sigma_i^*-\sigma_i)
(\sigma_i^*-\sigma_i)^T\big|\sigma,I=i\right]
&=\frac{1}{3}Id-
\frac{c_i}{3}(r_i\sigma_i^T+\sigma_ir_i^T)+
\sigma_i\sigma_i^T+R_i',
\end{split}\end{equation}
where
\[R_i'=\left(\frac{g(c_i)}{c_i}-\frac{1}{3}\right)Id-\left(g(c_i)-\frac{c_i}{3}\right)
\left(r_i\sigma_i^T+\sigma_ir_i^T\right).\]
Ignoring the $R_i'$ for the moment and putting the main term of 
\eqref{E-Delta^2-approx-crit} into \eqref{Delta^2-from-def-crit} yields
\begin{equation*}\begin{split}
\frac{c_3^2}{n^4}\sum_i\sum_{j,k\neq i}&\left[
\frac{1}{3}\inprod{\sigma_j}{\sigma_k}-
\frac{c_i}{3}(\sigma_j^Tr_i\sigma_i^T\sigma_k+\sigma_j^T\sigma_ir_i^T
\sigma_k)+\sigma_j^T\sigma_i\sigma_i^T\sigma_k\right].
\end{split}\end{equation*}
The first term is
\[\frac{c_3^2}{3n^4}\sum_i\sum_{j,k\neq i}\inprod{\sigma_j}{\sigma_k}=
\frac{c_3^2}{3n^4}\sum_{i=1}^n|\sigma^{(i)}|^2.\]
The first half of the second term is
\[-\frac{c_3^2}{3n^4}\sum_i\sum_{j,k\neq i}c_i\sigma_j^Tr_i\sigma_i^T\sigma_k
=-\frac{c_3^2}{3n^4}\sum_ic_i|\sigma^{(i)}|\inprod{\sigma^{(i)}}{\sigma_i}=-\frac{c_3^2}{n^5}\sum_i|\sigma^{(i)}|^2\inprod{\sigma^{(i)}}{\sigma_i},\]
and the second half is the same.
Finally, 
\[\frac{c_3^2}{n^4}\sum_i\sum_{j,k\neq i}\sigma_j^T\sigma_i\sigma_i^T
\sigma_k=\frac{c_3^2}{n^4}\sum_i\inprod{\sigma_i}{\sigma^{(i)}}^2;\]
all together,

\begin{equation*}\begin{split}
\E\left[(W_n'-W_n)^2\big|\sigma\right]&=\frac{c_3^2}{n^4}\sum_i
\left[\frac{1}{3}|\sigma^{(i)}|^2-\frac{2}{n}
|\sigma^{(i)}|^2\inprod{\sigma_i}{\sigma^{(i)}}+\inprod{\sigma_i}{
\sigma^{(i)}}^2\right]\\&\qquad\qquad
+\frac{c_3^2}{n^4}\sum_i\sum_{j,k\neq i}\sigma_j^TR_i'\sigma_k\\
&=\frac{c_3^2}{n^4}\sum_i
\left[\frac{2}{3}|\sigma^{(i)}|^2-\frac{2}{n}
|\sigma^{(i)}|^2\inprod{\sigma_i}{\sigma^{(i)}}\right]
+\frac{c_3^2}{n^4}\sum_i\sum_{j,k\neq i}\sigma_j^TR_i'\sigma_k\\&=
\frac{2c_3^2}{3n^3}\left(|S_n|^2-\frac{|S_n|^2}{n}+1\right)-\frac{2c_3^2}{n^5}\sum_i
|\sigma^{(i)}|^2\inprod{\sigma_i}{\sigma^{(i)}}+ \frac{c_3^2}{n^4}\sum_i\sum_{j,k\neq i}\sigma_j^TR_i'\sigma_k,
\end{split}\end{equation*}
where the computation for
$\E\left[\inprod{\sigma_i}{\sigma^{(i)}}^2\right]$ from the
supercritical case has been used.
Recall that the main term should be $kW_n=\frac{2c_3W_n}{3n^{3/2}}$ and
indeed it is.  It is a routine collection of arguments very similar to
those in the previous sections to show that the remaining terms are
bounded in expectation by $\frac{C\log(n)}{n^2}$.

Finally, part $(c)$ is straightforward as usual:
\begin{equation*}\begin{split}
\E|W_n'-W_n|^3=\frac{8c_3^3}{n^{9/2}}\E\left|\sum_{j\neq
    I}\inprod{\sigma_I^*-\sigma_I}{\sigma_j}\right|^3&=\frac{8c_3^3}{n^{9/2}}\E\left|\inprod{\sigma_I^*-\sigma_I}{S_n-\sigma_I}\right|^3\\&\le \frac{8c_3^3}{n^{9/2}}\left[8\E|S_n|^3+8\right]\le\frac{C\log(n)}{n^{9/4}}.\end{split}\end{equation*}

\end{proof}

\section{Appendix}
\subsection{Calculus of $\Phi_\beta$}
Recall that the free energy is obtained by minimizing the function 
 \begin{equation*}
\Phi_\beta(x) := \log\left(\frac{x}{\sinh(x)}\right)+x\coth(x)-1-\frac{\beta}{2}
\left(\coth(x)-\frac{1}{x}\right)^2. 
\end{equation*}
In the following lemmas, we explicitly identify the minima for all $\beta > 0$, and obtain an estimate used in the proof of Lemma 14.
\begin{lemma}\label{subcrit}
If $\beta\le 3$, then 
\[\inf_{x\ge 0}\left\{\log\left(\frac{x}{\sinh(x)}\right)+x\left(\coth(x)-
\frac{1}{x}\right)-\frac{\beta}{2}\left(\coth(x)-\frac{1}{x}\right)^2
\right\}=0,\]
achieved only at $x=0$.
\end{lemma}

\begin{proof}We show first that the expression to be minimized is increasing.
Differentiating the expression in question yields
\begin{equation*}\begin{split}
&\left[\frac{\sinh(x)}{x}\right]\left[\frac{\sinh(x)-x\cosh(x)}{\sinh^2(x)}
\right]+\left[\coth(x)-\frac{1}{x}\right]+x\left[\frac{1}{x^2}-\frac{1}{
\sinh^2(x)}\right]\\&\qquad\qquad\qquad-\beta\left[\coth(x)-\frac{1}{x}\right]
\left[\frac{1}{x^2}-\frac{1}{\sinh^2(x)}\right]\\
&=\left[\frac{1}{x^2}-\frac{1}{\sinh^2(x)}\right]\left[x-\beta
\left(\coth(x)-\frac{1}{x}\right)\right].
\end{split}\end{equation*}
Expanding $\sinh(x)$ in a Taylor series,
\[\sinh^2(x)=\left(x+\sum_{n=1}^\infty\frac{x^{2n+1}}{(2n+1)!}\right)^2> x^2.\]
The problem is therefore reduced to showing that
\[x-\beta
\left(\coth(x)-\frac{1}{x}\right)> 0,\]
or alternatively,
\[\beta<\frac{x^2}{x\coth(x)-1}.\]
Clearly showing this for $\beta=3$ suffices to prove the lemma.
Rearranging yet again, this is equivalent to showing that 
\[\coth(x)-\frac{1}{x}<\frac{x}{3}.\]
Expanding $\coth(x)$ in terms of $e^{2x}$ and rearranging terms, this
is furthermore equivalent to showing that
\[\left(1-x+\frac{x^2}{3}\right)e^{2x}>1+x+\frac{x^2}{3}.\]
Expanding $e^{2x}$ in a Taylor series, the left-hand side of the 
inequality above is given by 
\begin{equation*}\begin{split}
1+x+\frac{x^2}{3}+\sum_{n=3}^\infty x^n&\left[\frac{2^n}{n!}-\frac{2^{n-1}}{
(n-1)!}+\frac{2^{n-2}}{3(n-2)!}\right]\\&=
1+x+\frac{x^2}{3}+\sum_{n=3}^\infty \frac{2^{n-2}x^n}{3n!}\left[\left(n-\frac{7}{2}
\right)^2-\frac{1}{4}\right].\end{split}\end{equation*}
It is easy to see that the $n=3$ and $n=4$ terms in the power series above are
zero and that the rest are all
positive, thus completing the proof that the expression to be minimized
is increasing.  Moreover, recall that $\lim_{x\to0}
\coth(x)-\frac{1}{x}=0$ and $\lim_{x\to0}\frac{x}{\sinh(x)}=1$, 
so 
\[\lim_{x\to0}\left\{\log\left(\frac{x}{\sinh(x)}\right)+x\left(\coth(x)
-\frac{1}{x}\right)-\frac{\beta}{2}\left(\coth(x)-\frac{1}{x}\right)^2
\right\}=0.\]

\end{proof}

\begin{lemma}\label{supcrit}
For $\beta>3$, there is a unique value of $x\in(0,\infty)$
which minimizes
\[\log\left(\frac{x}{\sinh(x)}\right)+x\left(\coth(x)
-\frac{1}{x}\right)-\frac{\beta}{2}\left(\coth(x)-\frac{1}{x}\right)^2\]
over $[0,\infty)$.
\end{lemma}

\begin{proof}
From the previous proof, we have that the derivative of 
the expression to be minimized is
\[\left[\frac{1}{x^2}-\frac{1}{\sinh^2(x)}\right]\left[x-\beta
\left(\coth(x)-\frac{1}{x}\right)\right]\sim\frac{(3-\beta)x}{9},\]
for $x$ near zero.  For $\beta>3$, it follows that 
$x=0$ is a local maximum of the expression on $[0,\infty)$.  As
$x$ tends to infinity, the expression to be minimized is asymptotic to
$\log(x)$, and there is therefore at least one interior minimum.  
Since $\frac{1}{x^2}-\frac{1}{\sinh^2(x)}>0$, it must be the case that
at this interior minimum, 
\[x-\beta
\left(\coth(x)-\frac{1}{x}\right)=0;\]
that is, 
\[\beta=\frac{x}{\coth(x)-\frac{1}{x}}.\]
In fact, the function $g(x)=\frac{x}{\coth(x)-\frac{1}{x}}$ is strictly
increasing on $(0,\infty)$, and this equation thus uniquely determines
$x$ in terms of $\beta$.
First observe that
\[g'(x)=\frac{\coth(x)-\frac{2}{x}+\frac{x}{\sinh^2(x)}}{\left(\coth(x)
-\frac{1}{x}\right)^2},\]
and it thus suffices to show that 
\[\coth(x)-\frac{2}{x}+\frac{x}{\sinh^2(x)}>0\]
for $x>0$; multiplying through by $x\sinh^2(x)$, one could equivalently
show that
\[x\sinh(x)\cosh(x)+x^2-2\sinh^2(x)>0.\]
Using the identities $\sinh(x)\cosh(x)=\frac{\sinh(2x)}{2}$ and 
$\sinh^2(x)=\frac{\cosh(2x)-1}{2}$, this is equivalent to showing
that
\[\frac{x}{2}\sinh(2x)+x^2-\cosh(2x)+1>0.\]
Expanding the left-hand side in Taylor series yields
\[\sum_{n=2}^\infty x^{2n}\left[\frac{2^{2n-1}}{2(2n-1)!}-\frac{2^{2n}}{(2n)!}
\right]=\sum_{n=3}^\infty \frac{x^{2n}2^{2n-2}}{(2n)!}[2n-4],\]
all of whose terms are indeed positive.

\end{proof}

\begin{lemma}\label{beta_k_2}
Let $k_2$ denote the unique value of $x\in(0,\infty)$ with
\[x-\beta\left(\coth(x)-\frac{1}{x}\right)=0.\]
Then 
\[\beta\left(\frac{1}{k_2^2}-\frac{1}{\sinh^2(k_2)}\right)<1.\]
In particular, if 
\[\Phi_\beta(x):=\log\left(\frac{x}{\sinh(x)}\right)+x\left(\coth(x)
-\frac{1}{x}\right)-\frac{\beta}{2}\left(\coth(x)-\frac{1}{x}\right)^2,\]
then $\Phi_\beta'(k_2)=0$ and $\Phi_\beta''(k_2)> 0$.

\end{lemma}
\begin{proof}
Let 
\[f(x):=x-\beta\left(\coth(x)-\frac{1}{x}\right),\] 
so that
$f(k_2)=0$; as was shown in the previous proof, this uniquely defines
$k_2$ in terms of $\beta$.  Moreover, it was also shown that
$\lim_{x\to0}f(x)=0$, $\lim_{x\to\infty}f(x)=\infty$, and $\lim_{x\to0}
f'(x)<0$.  That is, $f$ is initially decreasing from 0, and then becomes
increasing eventually, crossing the $x$-axis exactly once 
in $(0,\infty)$.  It must therefore be that there is an $x<k_2$
such that $f'(x)=0$.  
  Now, 
\[f'(x)=1-\beta\left(\frac{1}{x^2}-\frac{1}{\sinh^2(x)}\right).\]
In fact, $g(x):=\frac{1}{x^2}-\frac{1}{\sinh^2(x)}$ is decreasing
on $(0,\infty)$: observe first that
\[g'(x)=2\coth(x)\csch ^2(x)-\frac{2}{x^3},\]
and so the claim is true if $\coth(x)\csch^2(x)<x^{-3}$.  By the 
definitions of the hyperbolic trigonometric functions, this is equivalent
to 
\[4x^3(e^{4x}+e^{2x})<(e^{2x}-1)^3.\]
Expanding in power series, this is equivalent to 
\begin{equation*}\begin{split}
\sum_{n=3}^\infty\frac{4(4^{n-3}+2^{n-3})}{(n-3)!}x^n&<
\left(\sum_{n=1}^\infty\frac{(2x)^n}{n!}\right)^3=
\sum_{n=3}^\infty\left(\sum_{\substack{i,j,k\ge 1\\i+j+k=n}}\frac{2^n}{i!j!k!}\right)
x^n.\end{split}\end{equation*}
Letting $\alpha:=i-1$, $\beta:=j-1$, and $\gamma=k-1$, the coefficient
of $x^n$ on the right-hand side is
\[\frac{2^n}{(n-3)!}\sum_{\substack{\alpha,\beta,\gamma\ge 0\\\alpha+\beta+
\gamma=n-3}}\binom{n-3}{\alpha,\beta,\gamma}=\frac{2^n}{(n-3)!}3^{n-3}.\]
It is now easy to see that the coefficient of $x^n$ on the right-hand
side is smaller than the one on the left for each $n\ge 3$, and so it
is in fact true that $g(x):=\frac{1}{x^2}-\frac{1}{\sinh^2(x)}$ is
decreasing on $(0,\infty)$.  It follows that $f'(x)$ is increasing
on $(0,\infty)$, and so the previously identified $x<k_2$ such 
that $f'(x)=0$ is in fact the only zero of $f'$, and $f'(k_2)>0$.

\medskip

Now, recall from the proof of the previous lemma that 
\[\Phi_\beta'(x)=\left[\frac{1}{x^2}-\frac{1}{\sinh^2(x)}\right]\left[x-\beta
\left(\coth(x)-\frac{1}{x}\right)\right].\]
The value $k_2$ was in fact determined by the fact that $\Phi_\beta'(k_2)=0$.
Moreover,
\begin{equation*}\begin{split}
\Phi_\beta''(x)&=\left[\frac{1}{x^2}-\frac{1}{\sinh^2(x)}\right]\left[1-\beta
\left(\frac{1}{x^2}-\frac{1}{\sinh^2(x)}\right)\right]\\&\qquad\qquad+
\left[2\coth(x)\csch ^2(x)-\frac{2}{x^3}\right]\left[x-\beta
\left(\coth(x)-\frac{1}{x}\right)\right],\end{split}\end{equation*}
and so
\[\Phi_\beta''(k_2)=\left[\frac{1}{k_2^2}-\frac{1}{\sinh^2(k_2)}
\right]\left[1-\beta\left(\frac{1}{k_2^2}-\frac{1}{\sinh^2(k_2)}\right)\right]
>0\]
\end{proof}

\subsection{Proof of Theorem \ref{T:abstract_approx}}\label{Stein_app}
This section is devoted to the proof of the abstract approximation
theorem used in Section \ref{S:critical}.  
For convenience, we recall the statement.
\newtheorem*{T:abstract_approx}{Theorem \ref{T:abstract_approx}}
\begin{T:abstract_approx}
Let $(W,W')$ be an exchangeable pair of positive random variables.
Suppose there exists a $\sigma$-field $\mathcal{F}\supseteq\sigma(W)$
and $k>0$  deterministic such that
\[\E\big[W'-W\big|\mathcal{F}\big]=3k\big(1-cW^2\big)+E\]
and
\[\E\big[(W'-W)^2\big|\mathcal{F}\big]=kW+E'.\]
Let $X$ have density 
\[p(t)=\begin{cases}\frac{1}{z}x^5e^{-\frac{ct^2}{2}}&t\ge
  0;\\0&t<0.\end{cases}\]
Then there are constants $C_1,C_2,C_3$  depending only on $c$ such that for all $h\in C^2(\R)$,
\begin{equation*}\begin{split}
\big|\E h(W)-\E h(X)\big|&\le
\frac{C_1\|h\|_\infty}{k}\E|E|+\left(\frac{C_2(\|h\|_\infty+
\|h'\|_\infty)}{k}\right)\E|E'|\\&\qquad+
\left(\frac{C_3(\|h\|_\infty+\|h'\|_\infty+
\|h''\|_\infty)}{k}\right)\E|W'-W|^3.\end{split}\end{equation*}
\end{T:abstract_approx}

In any version of Stein's method, a crucial component is the
characterization of the distribution of interest by a linear
operator.  The following lemma identifies the characterizing operator
for the random variable $X$ defined above.

\begin{lemma}\label{L:char}
Let $Y$ be a positive random variable.  Then $Y$ has density
\[p(t)=\begin{cases}\frac{1}{z}t^5e^{-3ct^2}&t\ge
  0;\\0&t<0.\end{cases}\]
if and only if 
\begin{equation}\label{char-eq}\E\big[Yf'(Y)+6(1-cY^2)f(Y)\big]=0\end{equation}
for all $f\in C^1((0,\infty))$ such that $\int_0^\infty f(t)t^5e^{-3ct^2}dt<\infty$.
That is, the characterizing operator $T_p$ for the distribution with density
$p$ is defined by
\[[T_p f](x)=xf'(x)+6(1-cx^2)f(x).\]
\end{lemma}

Not only is the random variable $X$ characterized by the operator
above, but this operator is invertible on $\{h:\E h(X)=0\}$, and the
inverse has the following important boundedness properties.

\begin{lemma}\label{boundedness}
Let $h:\R\to\R$ be given.  Suppose that
\[f=f_h:=\frac{1}{tp(t)}\int_0^t\big[h(s)-\E
    h(X)\big]p(s)ds,\]
with $p$ as above.  Then $[T_pf_h](x)=h(x)-\E h(X)$ and 
\begin{enumerate}
\item $\|f_h\|_\infty\le5\|h\|_\infty.$
\item
  $\|f_h'\|_\infty\le42\sqrt{c}\|h\|_\infty+3\|h'\|_\infty.$
\item
  $\|f_h''\|_\infty\le C_1\|h\|_\infty+C_2\|h'\|_\infty+
C_3\|h''\|_\infty,$ where $C_1,C_2,C_3$ are constants depending only
on $c$.
\end{enumerate}
\end{lemma}

With these two lemmas, the proof of Theorem \ref{T:abstract_approx} is relatively straightforward.
\begin{proof}[Proof of Theorem \ref{T:abstract_approx}]

Given $h$, let $f$ be the solution to the Stein equation described
above.  Then by exchangeability and the conditions on $(W,W')$,
\begin{equation*}\begin{split}
0&=\E[(W'-W)(f(W')+f(W))]\\&=\E[(W'-W)(f(W')-f(W))+2(W'-W)f(W)]\\
&=\E[(W'-W)^2f'(W)+E''+6k(1-cW^2)f(W)+2Ef(W)]\\
&=\E[kWf'(W)+E'f'(W)+E''+6k(1-cW^2)f(W)+2Ef(W)].
\end{split}\end{equation*}

Then
\[\E[Wf'(W)+6(1-cW^2)f(W)]=-\frac{1}{k}\E[E'f'(W)+2Ef(W)+E''],\]
and
\[|E''|\le \frac{\|f''\|_\infty}{2}|(W'-W)|^3.\]
The result is thus immediate from Lemma \ref{boundedness}.

\end{proof}

We conclude by giving the proofs of the key lemmas.
\begin{proof}[Proof of Lemma \ref{L:char}]
If $Y$ has the density above, then the fact that $Y$ satisfies
\eqref{char-eq} is just integration by parts.


For the reverse implication, one must solve the so-called Stein
equation; i.e., given $h:(0,\infty)\to\R$, find $f=f_h$ such that
\[tf'(t)+6(1-ct^2)f(t)=h(t)-\E h(X),\]
where $X$ has the density above.  The solution $f$ is given by 
\begin{equation*}\begin{split}f(t)&=\frac{1}{tp(t)}\int_0^t\big[h(s)-\E
    h(X)\big]p(s)ds \\&=-\frac{1}{tp(t)}\int_t^\infty \big[h(s)-\E h(X)\big]p(s)ds,\end{split}\end{equation*}
where $p(s):=\frac{s^5}{z}e^{-3cs^2}.$
Observe that 
\[\frac{d}{dt}\big[tf(t)p(t)]=[f(t)+tf'(t)]p(t)+tf(t)p'(t)=[h(t)-\E
h(X)]p(t),\]
so that
\[h(t)-\E
h(X)=f(t)+tf'(t)+\frac{tf(t)p'(t)}{p(t)}=6(1-ct^2)f(t)+tf'(t).\]
Here we have made use of the fact, frequently used below, that
\[\frac{tp'(t)}{p(t)}=5-6ct^2.\]

It is shown in the next lemma that $f$ and $f'$ are both bounded, and so if $Y$ satisfies \eqref{char-eq}, 
then if $h$ is given and $f$ solves the Stein equation, 
\[\E h(Y)-\E h(X)=\E\left[Yf'(Y)+6(1-cY^2)f(Y)\right]=0,\]
and so $Y\overset{d}{=}X.$
\end{proof}

\begin{proof}[Proof of Lemma \ref{boundedness}]
\begin{enumerate}
\item By the first expression for $f_h$, if $t\le t_o:=\sqrt{\frac{5}{6c}}$, then
\[f(t)\le\frac{2\|h\|_\infty}{tp(t)}\left(\int_0^tp(s)ds\right)\le\frac{\|h\|_\infty
  t^5}{3zp(t)}\le\frac{e^{5/2}\|h\|_\infty}{3}\le 5\|h\|_\infty.\]

By the second expression for $f_h$,
\[|f(t)|\le\frac{2\|h\|_\infty\P[X\ge t]}{tp(t)}.\]
It is easy to show directly that 
$\P[X\ge t]\le\frac{p(t)}{6ct}\left(1+\frac{2}{3ct^2}+\frac{2}{9c^2t^4}\right),$
and so 
\[|f(t)|\le\frac{2\|h\|_\infty}{6ct^2}\left(1+\frac{2}{3ct^2}+\frac{2}{9c^2t^4}\right)\le
\frac{84\|h\|_\infty}{125}\]
for $t\ge t_o$.  This completes the proof.

\item Recall that, because $f$ solves the Stein equation,
\[tf'(t)=6f(t)(ct^2-1)+h(t)-\E h(X).\]
For $t\le t_o$, observe that 
\begin{equation*}\begin{split}
h(t)-\E h(X)-6f(t)&=h(t)-\E h(X)-\frac{6}{tp(t)}\int_0^t[h(s)-\E
h(X)]p(s)ds\\&=\frac{6}{tp(t)}\int_0^t\Big([h(t)-\E h(X)]\frac{s^5p(t)}{t^5}-[h(s)-\E
h(X)]p(s)\Big)ds.
\end{split}\end{equation*}
Now, 
\begin{equation*}\begin{split}
&\left|\frac{1}{tp(t)}\int_0^t[h(t)-\E h(X)]\left(\frac{s^5p(t)}{t^5}-p(s)\right)ds\right|\\&\qquad\qquad\qquad\le
\frac{2\|h\|_\infty}{tp(t)}\int_0^t\left|1-\frac{s^5p(t)}{t^5p(s)}\right|
p(s)ds=\frac{2\|h\|_\infty}{tp(t)}\int_0^t\left|1-e^{3c(s^2-t^2)}\right|
p(s)ds\\&\qquad\qquad\qquad \qquad\qquad\qquad \qquad\qquad\qquad \qquad\,
\le2\|h\|_\infty ct^2,
\end{split}\end{equation*}
making use of the fact that $p(s)\le p(t)$ in this range.
Also,
\begin{equation*}\begin{split}
&\left|\frac{1}{tp(t)}\int_0^t\Big([h(t)-\E h(X)]-[h(s)-\E
  h(X)]\Big)p(s)ds\right|\\&\qquad\qquad\qquad\qquad\qquad\qquad\qquad\qquad\qquad
\le\frac{\|h'\|_\infty}{tp(t)}\int_0^t(t-s)p(s)ds\le
\frac{\|h'\|_\infty t}{2},
\end{split}\end{equation*}
so that for $t\le t_o$,
\[\frac{1}{t}\Big|h(t)-\E h(X)-6f(t)\Big|\le 12\|h\|_\infty ct_o+3\|h'\|_\infty,\]
and thus
\begin{equation*}\begin{split}
\big|f'(t)\big|&\le
6ct_o\|f\|_\infty+12\|h\|ct_o+3\|h'\|_\infty\\&\le42ct_o
\|h\|_\infty+3\|h'\|_\infty.\end{split}\end{equation*}
If $t\ge t_o$, then 
\[\big|f'(t)\big|\le 6ct|f(t)|+\frac{6\|f\|_\infty+2\|h\|_\infty}{t_o}.\]
By the second expression for $f$,
\[ct|f(t)|\le\frac{2c\|h\|_\infty\P[X\ge
  t]}{p(t)}\le\frac{\|h\|_\infty}{3t}\left(1+\frac{2}{3ct^2}+\frac{2}{9c^2t^4}\right),\]
making use again of the estimate $\P[X\ge t]\le\frac{p(t)}{6ct}\left(1+\frac{2}{3ct^2}+\frac{2}{9c^2t^4}\right).$
Taking $t_o=\sqrt{\frac{5}{6c}}$ as before and making some trivial
simplifying estimates completes the proof of this
part.

\item Differentiating both sides of the Stein equation gives that
\begin{equation}\label{f''}tf''(t)=12ctf(t)+(6ct^2-7)f'(t)+h'(t).\end{equation}
If $t\le t_o$, then $|12cf(t)|\le5c\|h\|_\infty $ and
by the Stein equation, $|6ctf'(t)|\le
36c\|f\|_\infty+12c\|h\|_\infty\le 192c\|h\|_\infty.$
Also by the Stein equation, 
\[f'(t)=6f(t)\left(ct-\frac{1}{t}\right)+\frac{h(t)-\E h(X)}{t}.\]
The first term can be absorbed into the existing bound, so that
\[|f''(t)|\le C\|h\|_\infty+\frac{1}{t}\left|h'(t)+\frac{42f(t)-7[h(t)-\E
  h(X)]}{t}\right|.\]
(From now on we will not bother to keep track of specific constants
and their dependence on $c$.)
Now, 
\begin{equation*}\begin{split}
h'(t)&-\frac{7[h(t)-\E
  h(X)-6f(t)]}{t}\\&=h'(t)-\frac{42}{t^2p(t)}\int_0^t\Big([h(t)-\E
h(X)]\frac{s^5p(t)}{t^5}-[h(s)-\E h(X)]p(s)\Big)ds\\&=-\frac{42}{t^2p(t)}\int_0^t\Big([h(t)-\E
h(X)]\frac{s^5p(t)}{t^5}-[h(s)-\E h(X)]p(s)-\frac{(t-s)s^5h'(t)p(t)}{t^5}\Big)ds\\&=-\frac{42}{t^2p(t)}\int_0^t\Big([h(t)-\E
h(X)]e^{-3ct^2}-[h(s)-\E h(X)]e^{-3cs^2}-(t-s)h'(t) e^{-3ct^2}\Big)\frac{s^5}{z}ds.
\end{split}\end{equation*}
Let $H(t):=[h(t)-\E h(X)]e^{-3ct^2}.$  Then 
\[H'(t)=\big[h'(t) -6ct[h(t)-\E
h(X)]\big]e^{-3ct^2},\]
and so the equation above becomes
\begin{equation*}\begin{split}
h'(t)&-\frac{7[h(t)-\E  h(X)-6f(t)]}{t}\\&=\frac{42}{t^2p(t)}\int_0^t
\Big(H(s)-H(t)-(s-t)H'(t)+6ct(t-s)[h(t)-\E h(X)]e^{-3ct^2}\Big)\frac{s^5}{z}ds.
\end{split}\end{equation*}
It follows that
\[\left|h'(t)-\frac{7[h(t)-\E
    h(X)-6f(t)]}{t}\right|\le\frac{42}{t^2p(t)}\left[\left(\sup_{s\in(0,t)}|H''(s)|\right)\left(\frac{t^8}{12z}\right)+\frac{ct^8}{z}[h(t)-
\E h(X)]e^{-3ct^2}\right].\]
Now, 
\begin{equation*}\begin{split}
H''(s)&=\left[h''(s)-12csh'(s)-6c(6cs^2-1)[h(s)-\E
  h(X)]\right]e^{-3cs^2},
\end{split}\end{equation*}
so
\[\sup_{s\in(0,t)}|H''(s)|\le\|h''\|_\infty+12ct_o\|h'\|_\infty+24c\|h\|_\infty.\]
All together, this gives that for $t\le t_o$,
\[|f''(t)|\le C_1
\|h\|_\infty +C_2\|h'\|_\infty+
\|h''\|_\infty,\]
where $C_1,C_2,C_3$ are constants depending only on $c$.

For $t>t_o$, it follows from \eqref{f''} and estimates already carried
out that
\begin{equation*}\begin{split}|f''(t)|&\le
    12c\|f\|_\infty+\left(6ct+\frac{7}{t_o}\right)|f'(t)|+\frac{\|h'\|_\infty}{t_o}\\
&\le
C_4\|h\|_\infty+C_5\|h'\|_\infty+ct|f'(t)|.\end{split}\end{equation*}

By the Stein equation, 
\[ctf'(t)=6cf(t)(ct^2-1)+ch(t)-c\E h(X),\]
and 
\[|c^2t^2f(t)|\le\frac{2c^2t\|h\|_\infty\P[X\ge
  t]}{p(t)}\le\frac{c\|h\|_\infty}{3}\left(1+\frac{2}{3ct^2}+\frac{2}{9c^2t^4}\right),\]
and so finally
\[|f''(t)|\le C_6\|h\|_\infty+C_4\|h'\|_\infty.\]

\end{enumerate}

\end{proof}

\section{Acknowledgments}
The authors thank Richard Ellis, Giovanni Gallavotti, and Enzo Marinari for helpful discussions.

\end{document}